\documentclass[a4paper,UKenglish,cleveref, autoref, thm-restate]{lipics-v2021-mod}
\pdfoutput=1 
\hideLIPIcs 
\nolinenumbers

\usepackage{fontawesome5}

\usepackage[normalem]{ulem}
\usepackage{svg}
\usepackage{xr}
\usepackage{graphicx}
\usepackage{dcolumn}
\usepackage{tabularx}
\usepackage{xcolor}
\usepackage{bm}
\usepackage{amsmath}
\usepackage{mathtools}
\usepackage{amssymb}
\usepackage{float}
\usepackage{algorithm}
\usepackage{algpseudocode}
\usepackage{tabularx}
\usepackage{pifont}
\usepackage{hyperref}
\newcommand{\cmark}{\ding{51}} 
\newcommand{\xmark}{\ding{55}} 
\usepackage{braket}
\usepackage{amsthm}
\usepackage{adjustbox}
\usepackage[table]{xcolor}
\usepackage{makecell} 
\usepackage{enumitem}

\hypersetup{breaklinks=true, colorlinks=true, linkcolor={blue!90!black}, citecolor={blue!90!black}, urlcolor={blue!90!black}}

\newcommand\Warning{\faExclamationTriangle}

\newcommand{\SK}[0]{\operatorname{SK}}
\newcommand{\PK}[0]{\operatorname{PK}}

\newcommand{\gen}[0]{\mathsf{Gen}}
\newcommand{\eval}[0]{\mathsf{Eval}}

\newcommand{\VRF}[0]{\operatorname{VRF}}
\newcommand{\VRFgen}[1]{\operatorname{VRF.Gen(#1)}}
\newcommand{\VRFeval}[1]{\operatorname{VRF.Eval(#1)}}
\newcommand{\VRFvfy}[1]{\operatorname{VRF.Vfy(#1)}}

\newcommand{\VRS}[0]{\mathsf{VRS}}
\newcommand{\vrs}[0]{\VRS}

\newcommand{\hash}[0]{\mathsf{Hash}}
\newcommand{\measure}[0]{\mathsf{Measure}}

\newcommand{\idealrandom}[0]{\mathsf{VRS}}
\newcommand{\idealblockchain}[0]{\mathsf{B3}_{\timebound, \bcentropybound}}

\newcommand{\idealseed}[0]{\mathsf{Seed}}

\newcommand{\blockheight}[0]{\mathsf{BID}}
\newcommand{\transactionhash}[0]{\mathsf{TXID}}

\newcommand{\aux}[0]{\mathsf{aux}}
\newcommand{\context}[0]{\mathsf{CTX}}

\newcommand{\ret}[0]{\mathbf{return}\ }
\newcommand{\ok}[0]{\mathsf{ok}}
\newcommand{\call}[0]{\mathsf{call}\ }

\newcommand{\timebound}[0]{\mathfrak{T}}
\newcommand{\bcentropybound}[0]{\mathfrak{E}}
\newcommand{\timeboundsec}[0]{\blockheight_\tau}
\newcommand{\timeboundtheo}[0]{\tau_{res}}

\newcommand{\simulator}[0]{\mathbf{\sigma}}
\newcommand{\hybrid}[0]{\mathbf{H}}

\newcommand{\cmdposttx}[0]{\mathsf{Post}}
\newcommand{\cmdread}[0]{\mathsf{Read}}
\newcommand{\cmdmine}[0]{\mathsf{Mine}}

\newcommand{\idealcommchannel}[0]{\mathsf{AUT}}
\newcommand{\cmdsend}[0]{\mathsf{Send}}
\newcommand{\cmdreceive}[0]{\mathsf{Rec}}
\newcommand{\cmdlatest}[0]{\mathsf{Latest}}

\newcommand{\cmdchallenge}[0]{\mathsf{Chal}}
\newcommand{\cmdresponse}[0]{\mathsf{Res}}

\newcommand{\cmdeval}[0]{\mathsf{Eval}}

\newcommand{\derive}[0]{\mathsf{Derive}}

\newcommand{\codecomment}[1]{\textcolor{gray}{// \text{#1}}}

\newcommand{\SR}[1]{\textcolor{red}{[SR: #1]}}
\newcommand{\AAA}[1]{\textcolor{blue}{[AAA: #1]}}
\newcommand{\pn}[1]{\textcolor{purple}{PN: #1}}
\newcommand{\ys}[1]{\textcolor{cyan}{YS: #1}}
\newcommand{\yeoh}[1]{\textcolor{orange}{Yeoh: #1}}
\newcommand{\shalti}[1]{\textcolor{pink}{Shalti: #1}}

\renewcommand{\SR}[1]{\errmessage{Error: \noexpand is disabled}}
\renewcommand{\AAA}[1]{\errmessage{Error: \noexpand is disabled}}
\renewcommand{\pn}[1]{\errmessage{Error: \noexpand is disabled}}
\renewcommand{\ys}[1]{\errmessage{Error: \noexpand is disabled}}
\renewcommand{\yeoh}[1]{\errmessage{Error: \noexpand is disabled}}
\renewcommand{\shalti}[1]{\errmessage{Error: \noexpand is disabled}}

\usepackage{thmtools}
\usepackage{thm-restate}

\newtheorem{assumption}{Assumption}

\usepackage[capitalise]{cleveref}
\crefname{lemma}{Lemma}{Lemmas}
\Crefname{lemma}{Lemma}{Lemmas}
\crefname{claim}{Claim}{Claims}
\Crefname{claim}{Claim}{Claims}

\usepackage[framemethod=TikZ]{mdframed}
\usepackage{xcolor}
\usepackage{environ}
\usepackage{varwidth}

\definecolor{titleBg}{RGB}{200, 220, 255}
\definecolor{frameColor}{RGB}{0, 0, 0}

\newlength{\contentwidth}

\NewEnviron{titledbox}[1]{%
    \setbox0=\hbox{%
        \begin{varwidth}{50cm}%
            \BODY
        \end{varwidth}%
    }%
    \setlength{\contentwidth}{\dimexpr\wd0+20pt\relax}%
    \begin{mdframed}[
        userdefinedwidth=\contentwidth,
        align=center,
        linewidth=1pt,
        linecolor=frameColor,
        backgroundcolor=white,
        roundcorner=0pt,
        innertopmargin=15pt,
        innerbottommargin=10pt,
        innerleftmargin=10pt,
        innerrightmargin=0pt,
        singleextra={%
            \node[
                fill=titleBg,
                draw=none,
                inner sep=5pt,
                rounded corners=2pt,
                anchor=west,
                xshift=10pt
            ] at (P-|O) {\bfseries #1};%
        },
        firstextra={%
            \node[
                fill=titleBg,
                draw=none,
                inner sep=5pt,
                rounded corners=2pt,
                anchor=west,
                xshift=10pt
            ] at (P-|O) {\bfseries #1};%
        },
    ]
        \BODY
    \end{mdframed}%
}

\newcommand{\mlxeb}{\mathsf{MLXEB}}
\newcommand{\xebtest}{\mathsf{XEB_{score}}}
\newcommand{\xebtestbool}{\mathsf{TestXEB}}
\newcommand{\qcam}{\mathsf{QCAM}}
\newcommand{\llha}{\mathsf{LLHA}}
\newcommand{\llqsv}{\mathsf{LLQSV}}
\newcommand{\qcamtime}{\mathsf{QCAMTIME}}
\newcommand{\trace}{\mathsf{Tr}}
\newcommand{\extract}{\mathsf{Ext}}

\usepackage[operators,sets,asymptotics, probability, advantage, adversary,keys,logic, primitives, ff, lambda]{cryptocode} 

\title{Verifiable Random Sampling}

\date{\today}

\author{Yeoh Wei Zhu}{Global Technology Applied Research, JPMorganChase}{}{}{}
\author{Soorya Rethinasamy}{Global Technology Applied Research, JPMorganChase}{}{}{}
\author{Anthony Alexiades Armenakas }{Global Technology Applied Research, JPMorganChase}{}{}{}
\author{Yash Satsangi }{Global Technology Applied Research, JPMorganChase}{}{}{}
\author{Shaltiel Eloul }{Global Technology Applied Research, JPMorganChase}{}{}{}
\author{Ruslan Shaydulin}{Global Technology Applied Research, JPMorganChase}{}{}{}

\authorrunning{Y. Wei Zhu et al.}

\begin{document}

\maketitle

\begin{abstract}
Verifiable random functions (VRF) underpin a wide range of applications that require publicly verifiable evaluations of a pseudorandom function on a given input. However, once the public key is published, the induced function is fixed and is a deterministic function of the input. This determinism can enable collusion and grinding-style attacks in which adversaries precompute and selectively exploit favorable input-output pairs. To address these limitations, we introduce the formal notion of verifiable random sampling (VRS). We propose a concrete VRS construction based on random quantum circuit sampling (RCS) executable on today's quantum computing devices. VRS supports multiparty protocols in which the verifier’s final output is a sample that is statistically close to a specified target distribution, while remaining publicly verifiable. We model the construction and prove its security within the constructive cryptography (CC) framework, thereby ensuring composability with other cryptographic protocols. Overall, our results provide a mechanism for verifiable random sampling that simultaneously guarantees sample freshness and public verifiability, enabling applications that require unpredictable, fresh randomness while preserving fairness through public verifiability.
\end{abstract}

\section{Introduction}
Many deployed applications rely on randomness to function correctly. 
Traditionally, random numbers are generated by extracting entropy from physical sources such as thermal noise \cite{Turan2018}, clock jitter \cite{Fischer2003}, and others \cite{Pironio2010}. For example, the Linux kernel aggregates entropy from keyboard interrupts, disk I/O, mouse movement, and device-driver interrupts \cite{Gutterman2006} to seed the randomness exposed through \texttt{/dev/random}. However, many modern applications operating in a cloud or distributed setting, such as smart contracts, have no access to a trusted local entropy source and must obtain randomness from a remote party over the Internet. In this setting, a client receiving a purported random sample has no way to verify two properties: (i) that the sample was actually drawn from the claimed physical source rather than fabricated by an adversarial provider, and (ii) that the sample is genuinely \emph{fresh}, i.e.,\ not precomputed. We address the problem of designing a primitive that provides both guarantees in a publicly verifiable manner.

A commonly used solution to this problem is the verifiable random function (VRF) \cite{micali1999verifiable}, which extends pseudorandom functions to the public-key setting, enabling anyone to verify that an output is computed correctly. VRFs have been adopted in a variety of applications, including e-lotteries \cite{lottery_vrf}, leader election in proof-of-stake protocols \cite{Kiayias2017,vrfalgorand}, and many others.
However, VRFs face fundamental limitations due to their deterministic structure: once the public key is fixed, the input-to-output mapping is fully deterministic. Consequently, the pseudorandomness guarantee depends on both the secrecy of the VRF provider's key and on non-collusion between the provider and the party choosing the input. If either condition fails, the output becomes predictable to the colluding parties and the guarantee is lost. Furthermore, even with an honest provider, the determinism enables grinding attacks in which a party who can influence the input precomputes many candidate evaluations and selectively submits the most favorable one. Since these limitations are structural, a suitable replacement must produce fresh outputs that are publicly verifiable, compose securely with other protocols, and support multiple use cases.

This motivates the following fundamental question: Can we construct a protocol that:
\begin{center}
\begin{minipage}{0.85\linewidth}
\itshape\raggedright

\begin{itemize}[leftmargin=1.6em, itemsep=0pt, topsep=0pt]
  \item outputs a sample from a target distribution $\mathcal{D}$,
  \item produces a publicly verifiable transcript,
  \item is composable, and
  \item is generic enough to support multiple use cases?
\end{itemize}
\end{minipage}
\end{center}

We answer this question in the affirmative by constructing a verifiable random sampling (VRS) protocol that builds on the RCS-based certified randomness primitive proposed in \cite{aaronson2023certified,Bassirian2026} and demonstrated experimentally in \cite{jpmc_cr, jpmc_cr_2}. In that primitive, a classical client samples random quantum circuits $\vec C := (C_1, \cdots, C_M)$ for some positive integer $M$ and sends them to a quantum provider, who evaluates $C_i\ket{0\cdots0}$ for all $i\in[M]$, measures in the computational basis, and returns the outcomes $\vec z:= (z_1,\cdots,z_M)$. Passing the XEB score test certifies that $\vec z$ carries genuine quantum entropy. In prior works, this primitive is an interactive two-party protocol whose guarantees rely on a private freshly sampled challenge and a tight response deadline measured on the verifier's local clock, neither of which is available to a third party observer of the transcript.

To turn this into a publicly verifiable sampling primitive, we replace the client's private circuit sampling with a public generation of the challenge circuits from the hash of the latest blockchain block, which is unpredictable to the provider before publication and recomputable by any third party afterwards. We replace the local timing deadline with the on-chain block interval $\timebound$ enforced by the consensus layer, and we require the response to be committed on-chain so that any third party can later reconstruct and verify the transcript. To make the protocol work in the blockchain setting, we rule out post-selection and stale-circuit reuse across clients through on-chain checks and domain separation. In addition, we split the classical client of \cite{jpmc_cr} into a challenge generator $T$ and a verifier $V$ that consumes the randomness, and we analyze the resulting three-party protocol in the constructive cryptography model. We provide a technical overview in \cref{sec:tech_overview} and summarize the contributions below.

\vspace{5mm}

\textbf{Contributions}
\begin{itemize}
    \item To the best of our knowledge, this is the first verifiable random sampling (VRS) formalization and construction $\idealrandom_{\mathcal{D}}^{\idealblockchain}$. Unlike VRFs, which yield a pseudorandom output, $\vrs_{\mathcal{D}}^{\idealblockchain}$ outputs a fresh random sample from a desired distribution $\mathcal{D}$.
    \item We construct $\idealrandom_{\mathcal{U}}^{\idealblockchain}$ based on random circuit sampling (RCS) that is implementable on quantum computing devices available today. Our model and framework support replacing the RCS-based entropy source with more efficient entropy sources in the future as fault-tolerant quantum computing devices become available.
    \item Given our VRS protocol construction $\idealrandom_{\mathcal{U}}^{\idealblockchain}$ for sampling from the uniform distribution $\mathcal{U}$, we show how to perform VRS for an arbitrary distribution $\mathcal{D}$ using new results on rejection sampling with error source distribution when considering information theoretic randomness.
    \item Our analysis is in the constructive cryptography (CC) model, which is composable. Our construction can therefore be arbitrary composed with other protocols.
    
\end{itemize}

\subsection{Related Works and Comparison to Prior Art}

We start by reviewing various related works in the domain of random number generation and then compare them against VRS.

\noindent\textbf{Randomness Beacon (RB).} Randomness beacons, introduced by \cite{rabin1983transaction}, provide a public source of randomness that is intended to be unpredictable prior to release in periodic intervals. Works such as \cite{Cascudo2017} propose a more efficient MPC protocol to derive an efficient randomness beacon assuming an  honest majority while \cite{Ewa2017} provide a randomness beacon in the (t,n) threshold model albeit with a non-negligible failure probability against a Byzantine adversary. Recently, beacons based on randomness generated by quantum processes have been demonstrated~\cite{kelsey2019a,Kavuri2025}. However, such beacons lack general provability guarantees, i.e., how a randomness consumer can be guaranteed that the published randomness is indeed a fresh value assuming the beacon committee is fully corrupted. 

\noindent\textbf{Verifiable Random Function (VRF).} VRF \cite{micali1999verifiable} offers a strong notion of provability: given a public key and an input, anyone can verify that an output was correctly computed under the corresponding secret key, without learning the secret key itself through the use of zero-knowledge proof. Recently, \cite{Giunta2024} studies an unbiasable VRF under skewed output distribution induced by malicious key setting.  Despite solving the provability issue, VRF suffers from the non-collusion assumption and from a deterministic output on a given input.

\noindent\textbf{Physically Unclonable Function (PUF).} PUF \cite{Gassend2002} leverages manufacturing variability to derive device-specific responses to challenges and is often proposed as a hardware-based primitive for identification \cite{che2015puf} and entropy extraction \cite{o2004puf}. The PUF typically assumes specific hardware properties and that the function is unlearnable; however, as shown in \cite{Rhrmair2010}, machine-learning-based modeling attacks can break this assumption. Moreover, PUF suffers from the key enrollment problem \cite{Pour2020}, where we have to trust both the provisioning and the enrollment process.

\noindent\textbf{Certified Randomness (CR).} Certified randomness focuses on proving that a sequence of bits originated from a \textit{truly} random process. Several approaches to certified randomness from quantum devices have been proposed. Random Circuit Sampling (RCS) based protocols \cite{aaronson2023certified,jpmc_cr,Bassirian2026} leverage the classical intractability of sampling from random quantum circuits where scoring high on cross-entropy benchmarking (XEB) score serves as a proof of quantum-ness and passing the XEB score test provably implies entropy generation under certain complexity assumptions (e.g., LLHA \cite{aaronson2023certified}). RCS-based protocols have been demonstrated experimentally \cite{jpmc_cr, jpmc_cr_2} on today's noisy quantum devices. However, these protocols suffer from expensive verification. On the other hand, \cite{Yamakawa2024} introduced a protocol based on an NP-search problem from error-correcting codes, proven hard classically but solvable by polynomial-time quantum computer. While this approach offers a more efficient verification process, this approach is not experimentally realizable today because it demands quantum decoding capabilities that are out of reach of today's noisy quantum computers. Therefore, we construct our verifiable random sampling protocol on top of an existing certified randomness protocol that works on current-scale quantum computers, while the efficiency of our construction can be further improved in the future by switching to a certified randomness protocol designed for fault-tolerant quantum computers. 

\noindent\textbf{Verifiable Random Sampling $\vrs$ (This work).} 
$\vrs$ solves the provability issue of randomness beacon while simultaneously guaranteeing fresh randomness, unlike the stale pseudorandom function of  a VRF. Compared to PUF, $\vrs$ relies on the principles of quantum mechanics, requiring neither a trust assumption nor a hardware assumption. Notably, $\vrs$ differs from certified randomness, as it samples from any desired distribution in a publicly verifiable manner. 
The comparisons are summarized in \cref{tab:compare}.

\label{sec:comparisons}
\begin{table}[ht!]
\rowcolors{1}{}{lightgray}
\begin{adjustbox}{minipage=\linewidth,scale=1}
\centering
\begin{tabular}{r|r|r|r|r|r}
  \hline
 & \makecell{$\vrs^{\idealblockchain}_{\mathcal{U}}$ \\ (This work)} & VRF & PUF & RB & CR - RCS \\
  \hline
  Collusion Resistance & \checkmark & \xmark & \xmark & \xmark & \checkmark \\
  Publicly Verifiable & \checkmark & \checkmark & \checkmark & \xmark & \xmark \\
  Non-deterministic Output & \checkmark & \xmark & \checkmark & \checkmark  & \checkmark \\
  Information-theoretic Randomness & \checkmark & \xmark & \xmark & \checkmark*  & \checkmark \\
  Specified Distribution & \checkmark & \xmark & \xmark & \xmark  & \xmark \\
   \hline
\end{tabular}
\end{adjustbox}
\caption{Comparison of VRS against various random number generation related primitives. \\ * The scheme may exhibit the required property depending on the construction.} 
\label{tab:compare}
\end{table}

\section{Technical Overview}
\label{sec:tech_overview}

In this section, we present the high-level technical overview of our paper. 
We consider the use case of requiring (unpredictable) verifiable randomness for a generic party $V$ in this paper. The high level overview of the construction $\vrs^{\idealblockchain}_{\mathcal{U}}$ is given in \cref{fig:vrs_overview}.
\begin{figure}
    \centering
    \includegraphics[width=0.85\linewidth]{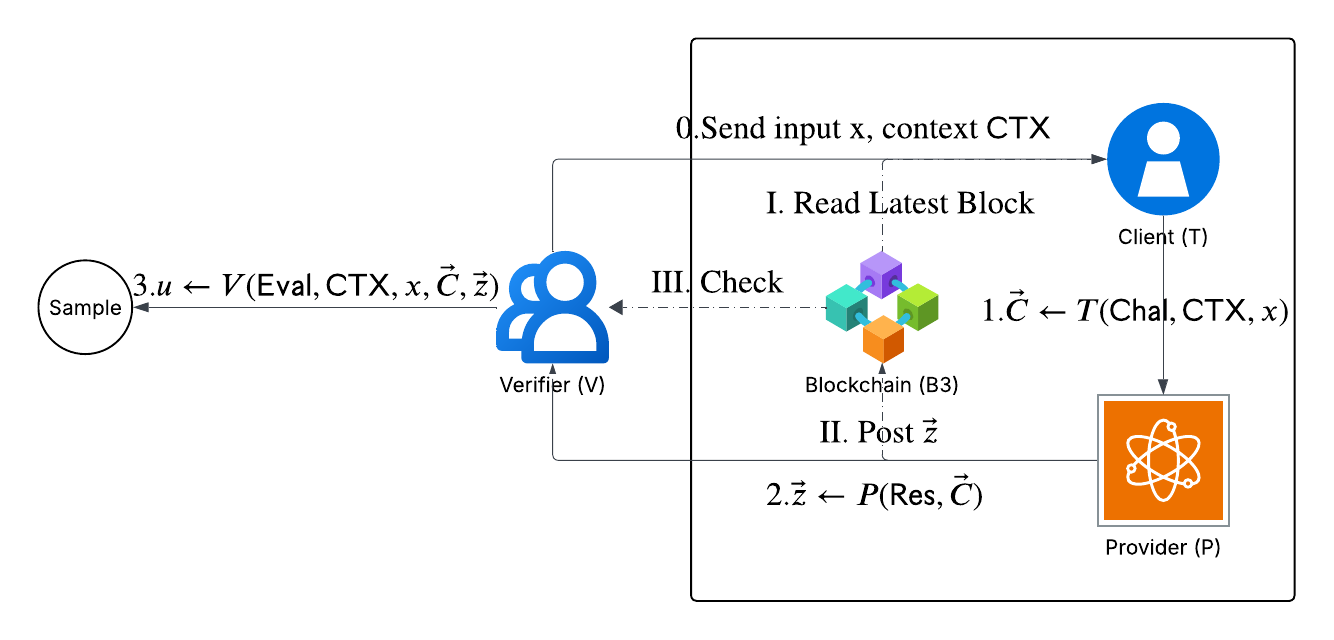}
    \caption{Simplified $\vrs_{\mathcal{U}}^{\idealblockchain}$ overview. Steps I, II, III are executed inside Steps 1, 2, 3 respectively. Step 0 is optional. Steps inside the box indicate online phase where timing is important. The syntax of $T,P,V$ is modified to include additional inputs for clarity.}
    \label{fig:vrs_overview}
\end{figure}

We first state the parties in the $\vrs$ protocol from \cref{sec:main_overview}:
\begin{itemize}
    \item Client $T$: Generates challenge circuit $\vec C$ using input $x$ and latest block data.
    \item Provider $P$: Provides the response $\vec z$ with some intrinsic entropy, and posts it to the blockchain.
    \item Verifier $V$: Verifies the response with respect to the circuit and outputs a sample $u$.
\end{itemize}

\subsection{(Provable) Random Number Generation: Transitioning from Deterministic Function to Non-Deterministic Sampling}

\textbf{Verifiable Random Function (VRF).} VRF remains the most common way to inject verifiable randomness into the execution of smart contracts. VRFs work by provably transforming the smart contract's input into an output that looks random enough for use cases such as blockchain-based fair matchmaking, fair ordering, lucky draws, and many others \cite{chainlink_usecase}. However, this model of supplying randomness critically relies on the non-collusion of the VRF provider and the party selecting the VRF input. Since the random function used in VRF is determined in advance (by publishing the public key), the colluding parties can indeed essentially ``predict'' the random function evaluation and gain advantage in Web3 on-chain games, lotteries and other applications. 

\textbf{Verifiable Randomness Source.} The deterministic nature of a fixed random function evaluation can be overcome by introducing a source that is probabilistic in nature. However, proving that a sample returned over the Internet was freshly drawn from some distribution is a non-trivial problem. To that end, Aaronson and Hung \cite{aaronson2023certified} showed that we can provably lower bound the min-entropy of the output string given by a potentially adversarial party's algorithm over the Internet. The protocol goes as follows. 
Let $M$ be the number of circuits used in the protocol. The protocol starts with a classical verifier generating $n$-qubit challenge circuits $\vec C = C_1, C_2, \ldots, C_M$ pseudorandomly, which are then sent to the provider. The provider's server then evaluates each circuit on the initial state $\ket{0^n}$ and measures the final state in the computational basis. The output string $\vec z = z_1,\ldots,z_M$ is sent back to the verifier as the response. In \cite{jpmc_cr}, the verifier then computes the XEB test score 
\begin{linenomath}
$$
\xebtest = \frac{N}{|\mathcal{V}|}\sum_{i \in \mathcal{V}} |\langle z_i \vert C_i \vert 0^n\rangle |^2-1,
$$
\end{linenomath}
where $N = 2^n$, $\mathcal{V}$ is the verification set, and the probability of observing a string $z_i$ is $p_i(z_i)=|\langle z_i \vert C_i \vert 0^n\rangle|^2$.
Assuming the hardness of spoofing the XEB test score, if the test score is sufficiently large $\xebtest \geq \chi$, then there must be some genuine entropy in the response string (See \cref{theo:smooth_min_entropy_bound_guarantee} for the formal statement). We provide an overview of the hardness justification in \cref{sec:hardness}. 

We briefly explain the intuition behind the hardness of spoofing the XEB score. For random circuits sampled from the Haar distribution, $C_i \sim Haar(N)$ (See \cref{sec:preliminaries} for the definition of the Haar measure), the probability distribution $\{p_i(z_i)\}_{z_i}$\footnote{The probability distribution approximates the Porter-Thomas distribution} is highly non-uniform, that is a large fraction of output strings have probability $p < 1/N$, and the probability density is concentrated on a relatively smaller fraction of the output strings. A genuine quantum device sampling from $p_i$ will therefore output strings that are biased toward these
higher-probability outcomes, while classical spoofing strategies, assuming the unstructured nature of RCS circuits from \cref{assum:rcs_unstructure}, will not be able to output enough of such higher probability output strings. 
The $\xebtest$ measures this bias by averaging the ideal probabilities $p_i(z_i)$ of the returned samples and rescaling by $N$. For classical spoofing attempts, it satisfies $\mathbb{E}[N\,p_i(z_i)]\approx 1$ and hence $\xebtest \approx 0$. Meanwhile, ideal sampling satisfies $\mathbb{E}[N\,p_i(z_i)]\approx 2$ and $\xebtest \approx 1$. Thus, a larger XEB score indicates stronger correlation with the intended circuit distribution.

Note here that a straightforward classical spoofing strategy is to simulate the entire RCS circuit, find the output distribution, and sample from it. This would lead to a high $\xebtest$ score. However, this type of simulation is computationally expensive even using state-of-the-art classical simulation methods. Thus, requiring that the response is submitted within a certain time-bound rules out such potential spoofing attacks \cite{aaronson2023certified}, and is essential to the protocol. We will make use of this verifiable randomness source in the next section to build Verifiable Random Sampling (VRS).

\subsection{Constructing $\vrs_{\mathcal{U}}$: From Certified Randomness to Verifiable Random Sampling (VRS).} 
\textbf{Extracting Uniform Sample from Randomness Source.} 
The output of the verifiable randomness source $Z$ is an output string $\vec z$ with a verifiable (smooth) min-entropy guarantee, such that $H^{\epsilon_h}_{\operatorname{min}}(Z) \geq B$ for some bound $B$. To extract uniform randomness, we assume a second source of (private) randomness. We can now apply the randomness extractor theorem (formally stated in \cref{def:randomness_extraction}). We specifically consider a quantum-proof strong extractor because the string $z$ is given by a (potentially malicious) quantum algorithm who may possess quantum side-information on it. The output distribution of the randomness extractor is then statistically close to the uniform distribution as stated in \cref{def:randomness_extraction}, even if one of the inputs is  publicly known.

\textbf{Public Verifiability.} The above output is not publicly verifiable, however. When a third party is given the transcript of the above interaction, they cannot be sure whether the interaction was completed on-time or whether the challenge circuits were honestly generated. We solve both of these issues by using a blockchain resource, which provides both an alternate entropy source to create the challenge circuits and a method to timestamp the interaction. Specifically, the latest blockchain hash data is used to generate the challenge circuits, which ensures that the interaction is linked to on-chain activity, serving as a third-party timestamp. The blockchain hash data contains enough computational entropy \cite{bitcoin_entropy} (though it is potentially biasable, see also \ref{sec:blockchain_fail}), which ensures that the challenge circuits are freshly generated. Furthermore, posting the responses to the chain is useful for enforcing timing constraints. Assuming that there exists a certain time bound $\timebound$ between each block publication such that
\begin{linenomath}
$$
\mathbf{B}[id+1].\tau - \mathbf{B}[id].\tau < \timebound,
$$
\end{linenomath}
where $\mathbf{B}$ denotes the blocks, and $\tau$ denotes the timestamp, we can be sure that the challenge-response protocol is completed within some time bound and can subsequently rule out the time-consuming spoofing or simulation attacks \cite{aaronson2023certified}. We note here that our protocol works with any bulletin-board-like scheme other than blockchain, as modeled by the ideal blockchain $\idealblockchain$ defined formally in \cref{sec:idealblockchain_formal}, as long as it provides a time-stamping feature at the granularity specified by $\timebound$, and broadcasts pseudorandom numbers at a regular interval that contains at least $\bcentropybound$-bits of entropy.

In addition to ensuring the timing constraint, we also enforce that no post-selection can be done by the client. A client might decide to launch several VRS protocols in parallel, post all the transcripts, and ``choose'' the best transcript. This would bias the output distribution. To prevent this type of manipulation, we as the third party itself must pick the transcript. For example, if the third-party supplied input $x$ (like $V$ in our protocol), it can either pick the earliest transcript, or pick among the available transcripts uniformly at random. Our verifier's algorithm will pick the first available transcript after a verifier specified time to ensure both that no pre-sampling is done before the verifier specified time and that no post-selection is done.

To further ensure public verifiability of the output sample $u$ by $V$ as a third-party observer outside the protocol against potentially malicious parties $T$, $P$, and $V$, the verifier further performs one of the following, depending on the context:
\begin{itemize}
    \item Non-colluding Verifier with Provider and One-time Use: The verifier attaches the second input to the randomness extractor alongside the sample $u$ together with the transcript.
    \item Non-colluding Verifier with Provider and Multiple Use: The verifier simply generates a NIZK proof of correct generation of a sample $u$ against a seed committed before the protocol execution.
    \item Corrupted Verifier: The verifier, instead of using private input, samples from the blockchain hash data of the block in which the response is published. Then, the verifier uses the sampled blockchain hash data to derive the final sample $u$. 
\end{itemize}

Note that for all three cases above, the extractor's second input remains secret until the moment the provider $P$ publishes its response. Now, the third party simply executes the same evaluation/verification algorithm $\pi_V$ executed by an honest verifier $V$.
We discuss these instantiations further in \cref{sec:trust_model_instant}.

\noindent \textbf{Ensuring Independent Inputs.} 
Randomness extractor theorems (see \cref{def:randomness_extraction}) require the independence of the primary and secondary inputs. For the randomness extractor used in the final sample generation, the other input is assumed to be private and is modeled as such in \cref{fig:ideal_random}, which means the extractor input string from the provider does not depend on this unknown second input.

\noindent \textbf{Domain Separation: Avoiding Challenge Circuit Staleness across Multiple Clients.} 
Multiple different clients might try to reuse the same input $x$ to derive their randomness. However, it can be seen that given the same input $x$ and the same blockchain hash, the same output string can be reused while passing the $\xebtestbool$ verification test, since the challenge circuit is the same. Note that it does not contradict the entropy guarantee from \cref{theo:smooth_min_entropy_bound_guarantee}, which requires the circuit to be (pseudo-)randomly generated on each invocation.
This issue can be solved by requiring the client (and verifier) to perform (and verify) domain separation in the input space for the hash function used to pseudorandomly derive the challenge circuit. Domain separation or salting in the hash function input space is a common security practice \cite{Bellare1993, Bellare2020, Davis2022, Krawczyk2010}. Note that we model this case using a generic domain separator $\context$, since the solution to this problem can be context or domain-specific as discussed above, and we discuss the domain separation in detail in \cref{sec:domain_seperation}. We assume that different contexts will be represented as different values in $\context$. For example, $\context$ can be a session ID.

\noindent \textbf{Protocol Summary. }
We give a $\vrs$ protocol high-level summary in \cref{sec:protocol_summary}.

\subsection{Use Cases and Extensions}
\subsubsection{Lifting $\vrs_{\mathcal{U}}$ to $\vrs_{\mathcal{D}}$. } 
We constructed a protocol $\vrs_{\mathcal{U}}$, given in \cref{sec:vrs_cc}, that allows one to sample from the uniform distribution. To sample from an arbitrary distribution $\mathcal{D}$, we can make use of the rejection sampling technique \cite{von195113}. The technique is built on the premise that perfect sampling from $\mathcal{U}$ is possible. However, it may be the case that we are only able to sample from $\tilde{\mathcal{U}}$ that is $\epsilon-$close to ${\mathcal{U}}$ in statistical distance. To alleviate this, we further derive a result on the rejection sampling with error stated as \cref{lemma:rejection_sampling_error} in \cref{sec:rejection_sampling_error}. The rejection sampling with error works by sampling from the imperfect distribution $\tilde{\mathcal{U}}$ and then rejects with probability as if we are sampling from the intended source distribution $\mathcal{U}$. With the rejection sampling with error, we are able to bound the statistical distance of the sample obtained from $\tilde{\mathcal{D}}$ with the intended output distribution $\mathcal{D}$, where $\tilde{\mathcal{D}}$ is the output distribution of the rejection sampling with error algorithm.
\subsubsection{Use Cases for $\vrs$}
\label{sec:usecases}

To demonstrate the effectiveness of the $\vrs$ formal model across various use cases with different trust assumptions, we showcase its applicability in the following scenarios:

\begin{enumerate}
    \item \textbf{Publicly Verifiable Randomness Derivation.} In a publicly verifiable randomness derivation use case, a verifier who may lack a local qualified randomness source would like to derive a random value by observing the transcript obtained from the interaction between a client $T$ and a provider $P$. Therefore, we have the setting where $P$ and $T$ are potentially malicious while $V$ is honest. This is the exact setting analyzed in \cref{sec:vrs_cc}. 
    \item \textbf{Smart Contract Random Value Derivation.} 
    Similarly, smart contract use cases typically fall into the above setting but with an honest $T$, since a smart contract usually lacks local randomness and behaves honestly. The setting in which both the client and verifier are honest (where the verifier also serves as the challenge-generating client) while the provider is potentially malicious is also addressed by the protocol defined in \cref{sec:vrs_cc}, since the capabilities of a malicious client are a strict superset of those of an honest one (a malicious party can always adopt the honest strategy). By substituting the second input (seed) in the construction with a future blockchain hash value that is unpredictable, the smart contract can likewise derive a statistically close sample.
    \item \textbf{e-Lottery: Verifiable Fairness for All.} 
    In an e-Lottery setting, a set of mutually distrusting clients $T$ and a lottery operator $LO$ wish to jointly compute a function $V$ that determines the winner. The winning criterion is that the party who draws the smallest sample from a uniform distribution wins. For simplicity, the seed can be defined as a shared secret $seed = \sum s_i$, jointly computed by each party $T_i$ and $LO$ over pre-published commitments $\mathsf{Commit}(s_i)$, where $s_i \in \mathbb{Z}_q$. A similar construction can be realized using threshold distributed key generation, as described in \cite{Gennaro2006}.
    In this setting, $T$ and $P$ are malicious, while $V$ is honest (enforced by the MPC protocol computing the function $V$), matching the construction given in \cref{sec:vrs_cc}.
\end{enumerate}

\section{Preliminaries}
\label{sec:preliminaries}

We use $a \sim \mathcal{D}$ to denote that $a$ is a random sample drawn from a distribution $\mathcal{D}$ and use $b \sample D$ to denote that $b$ is a random sample from a set $D$. Let $C$ be an $n$-qubit quantum circuit and let $N=2^n$ be the dimension of the Hilbert space $\mathcal{H}$. We define $\mathcal{P}(\mathcal{H})$ to denote the set of positive semi-definite operators on $\mathcal{H}$. The trace of an operator $\rho$ is given by $\trace(\rho) := \sum_i \braket{i|\rho|i}$ for any orthonormal basis $\{\ket{i}\}$. We can then define the normalized quantum states and sub-normalized quantum states as $\mathcal{S}(\mathcal{H}) := \{\rho \in \mathcal{P}(\mathcal{H}) : \trace( \rho) = 1\}$ and $\mathcal{S}_{\leq}(\mathcal{H}) := \{\rho \in \mathcal{P}(\mathcal{H}) : \trace( \rho) \leq 1\}$. In the interest of notational convenience, we use the tensor product notation $\vec C := C_1\otimes \cdots\otimes C_M$ to denote the composite circuit and consequently $\vec z$ to denote the output string obtained from measuring the state $\vec C\ket{0^{n\cdot M}}$ in the computational basis rather than addressing them separately via index notation. In practice, the circuit is represented using the individual factor $C_i$ rather than the full tensor product.

For RCS experiments, the relevant distribution is the distribution of bitstrings given a particular circuit $C$, denoted by $p_C(z) := |\langle z |C| 0^n\rangle|^2$. An alternate notation used in this work is using the $\measure$ operation, i.e., $\measure(\ket{\phi})$ is the random variable obtained by measuring the state $\ket{\phi}$ in the computational basis. 

We denote the hash function as $\hash(\cdot)$ and use the overloaded symbol  $\hash_{\mathcal{D}}(\cdot)$ to mean hashing into the specific distribution $\mathcal{D}$. Let $X,Y$ be two random variables with support $S_X, S_Y$ respectively, the statistical distance is defined as 
$ \mathsf{SD(X,Y) := \frac{1}{2}\sum_{u \in S_X \cup S_Y}} |\Pr[X=u] - \Pr[Y=u]|.$

\textbf{Haar-random Circuit.} Let $Haar(N)$ denote the unique bi-invariant probability measure on the unitary group $U(N)$ so that for any fixed unit vector $\ket{\phi} \in \mathbb{C}^N$ and $C \sim Haar(N)$, $C\ket{\phi}$ is uniformly distributed over the unit sphere in $\mathbb{C}^N$. For deriving a Haar-random circuit, it is known that with a uniform seed, we can sample from a distribution that approximates $Haar(N)$ in a pseudorandom manner \cite{Metger2024, Brando2016}. 
We denote such an algorithm $\derive_{Haar(N)^M}: \{0,1\}^s\rightarrow (C_1, \cdots, C_M)$ where $C_i$ is an $n$-qubit quantum circuit.

\textbf{Entropy.} For a quantum state $\rho$, $H$ denotes the von Neumann entropy $H(\rho):= - \trace (\rho \log \rho)$, which reduces to the Shannon entropy in the classical setting $H(X) := - \sum_{i} \Pr(x_i) \log \Pr(x_i)$ for some random variable $X$, when $\rho$ is diagonal. $H^{\epsilon}$ denotes the smooth version of the entropy measure \cite{tomamichel2009fully}. 
Given a classical-quantum state $\rho_{XA}$ classical on X, the conditional min-entropy is defined as $H_{\operatorname{min}}(X|A)_\rho :=- \log p_{\operatorname{guess}}(X|A)_{\rho}$ where $ p_{\operatorname{guess}}(X|A)_{\rho}$ is the optimal probability of guessing. The smooth version is defined using $\epsilon$-ball of states around $\rho \in \mathcal{S}(\mathcal{H})$ defined as $B^{\epsilon}(\rho) := \{\sigma \in \mathcal{S}_{\leq}(\mathcal{H}) : P(\rho,\sigma)\leq \epsilon\}$, where the purified distance metric is given by
$P(\rho,\sigma):= \sqrt{1-F^2(\rho,\sigma)}$, and $F(\rho,\sigma):=\norm{\sqrt\rho \sqrt \sigma}_1$ and $\norm{A}_1 = \trace({\sqrt{A^\dagger A}})$.
Given a smoothing parameter $\epsilon$, the smooth min-entropy can then be defined as $H_{\operatorname{min}}^{\epsilon}(X|A)_\rho := \sup_{\sigma_{XA} \in B^{\epsilon}(\rho)} H_{\operatorname{min}}(X|A)_{\sigma}$. 
In the classical setting, we use $H_{\operatorname{min}}(X)$ to denote the min-entropy for a random variable X such that $H_{\operatorname{min}}(X)=H_{\infty}(X)= - \log (\max_{x} \Pr[X=x])$. 

\subsection{Cryptographic Assumption}
\label{sec:crypto_assump}
We first recall the entropy guarantee that arises from complexity assumptions. We only state the theorem here without definition, see \cref{sec:hardness} for more details.

\begin{theorem}[Passing MLXEB test with low entropy solves $\llqsv$(Theorem 8, \cite{jpmc_cr})]
There exists a quantum-classical Arthur-Merlin protocol which on input of an $O(n)$-bit advice string solves $\llqsv_B(\mathcal{D})$ ,which means $\llqsv_B{(\mathcal{D})}\in \qcamtime(2^Bn^{O(1)})/O(n)$, if there exists a device $\mathcal{A}$ which runs in polynomial time and satisfies the following:
\begin{itemize}
    \item  $\mathcal{A}$ solves $\mlxeb$ with probability $q=\Pr_{\vec C \sim \mathcal{D}^k, \vec z \sim \mathcal{A}(\vec C)}\left[\sum_{i=1}^k p_{C_i}(z_i) \geq \frac{bk}{N}\right]$,
    \item $H(Z | \vec C )_{\mathcal{A}} < \frac{B}{2} \left(\frac{bq-1-\epsilon}{b-1}\right)$ where $\epsilon = n^{-O(1)}$.
\end{itemize}
\end{theorem}

$\llqsv$ is a conjectured hard problem in \cite{aaronson2023certified} and it is proved that $\llqsv(\mathcal{D}) \notin \qcamtime(\allowbreak 2^Bn^{O(1)})/q(2^{B}n^{O(1)})$ in the random oracle model.
Despite the above, it is currently hard to instantiate an experimentally viable protocol based on complexity-theoretic arguments. 
We instead rely on the following min-entropy bound obtained by using the assumption given below, which is supported by the complexity argument given in \cref{sec:hardness} as otherwise we could exploit the presumed structure to solve $\llqsv$ with the gap required: 

\begin{assumption}[Unstructuredness of Random Circuit Sampling and Hardness of $\xebtest$]
\label{assum:rcs_unstructure}
Output distribution of sufficiently deep RCS circuits is ``unstructured''. That is the classical description of the circuit does not leak any information about which output are heavy or light. Consequently, the best strategy for optimizing $\xebtest$ is the frugal rejection sampling strategy described in \cite{jpmc_cr} which approximately simulates the circuit classically and returns the output as it is. 
\end{assumption}

\begin{remark}
The above assumption is used in the Google quantum supremacy experiment \cite{Arute2019}. A similar assumption is used in \cite{aaronson2023certified, jpmc_cr_2, Bassirian2026} and the assumption is analyzed in \cite{Aaronson2020, Boixo2018}. 
\end{remark}

For convenient we define the XEB test as follow.
\begin{definition}[$\xebtestbool_{\alpha, \chi}$ \cite{jpmc_cr}]
\label{def:xebtest_bool}
Let $\mathcal{D}$ be a probability distribution over quantum circuits on n qubits. Given $\vec C := (C_1, \dots, C_M)$ drawn from $\mathcal{D}^k$, a sample $\vec z =(z_1,\dots,z_M) \in (\{0,1\}^n)^M$, a threshold score $\chi$ and a test set size $\alpha$, $\xebtestbool_{\alpha, \chi}(\vec C,\vec z)$ is defined as follow:
$$
\xebtestbool_{\alpha, \chi}(\vec C,\vec z)=
\begin{cases}
1, & \text{if } \frac{2^n}{|\mathcal{V}|} \sum_{i \in \mathcal{V}} p_{C_i}(z_i) \geq \chi +1 \text{ where } \mathcal{V} \sample \{v : v \subseteq [M] \wedge |v| =\alpha \},\\
0, & \text{otherwise}.
\end{cases}
$$
\end{definition}

Unstructuredness of RCS motivates the oracle-access model where the (classical) adversary is given a quantum computer resource as a query oracle that evaluates the quantum circuit honestly and returns the sampled output string. 
We follow Ref.~\cite{jpmc_cr}, which specifically studies the case where the adversarial algorithm has just enough classical computing power to spoof $M-Q_{\vec C}$ number of samples and has to honestly return at least $Q_{\vec C}$ samples from the quantum computer within a time bound $\timeboundtheo$. However, we remark that our analysis can be straightforwardly extended to the more general adversary of Ref.~\cite{jpmc_cr_2}, which we leave to future work.
We first recall the simplified one-shot protocol without batching from \cite{jpmc_cr} that proceeds as follows:
\begin{enumerate}
    \item The (challenge) client pseudorandomly generates $M$-numbered $n$-qubits circuits $\vec C:= C_1,\ldots,C_M$.
    \item For each $i \in [M]$, the provider evaluates the circuits $C_i\ket{0^n}$ and measures in computational basis to obtain $z_i$. The provider then returns $\vec z := z_1,\ldots,z_M$.
    \item The (verifier) client aborts if $\xebtestbool_{\alpha,\chi}(\vec C, \vec z)=0$ for some 
    $\alpha, \chi$, otherwise the verifier outputs the sample $\vec z$.
\end{enumerate}

Let $\epsilon_s \in (0,1/4)$. 
Let $\tilde I$ be the classical side information consisting of the seed $K_{seed}$ used to generate the challenge circuit $\vec C$ and some side information $S^0$. 
We now recall the min-entropy bound from \cite{jpmc_cr} but we include the necessary assumptions made by \cite{jpmc_cr} into the theorem statement.
\begin{theorem}[Entropy Guarantee, Theorem 1 \cite{jpmc_cr}]
\label{theo:smooth_min_entropy_bound_guarantee}
Assuming the protocol described above and the adversary making at most $Q_{\vec C}$ queries to the quantum circuits oracle $\mathcal{O}_{\vec C}$ defined below, where $\vec C \sim Haar(N)^{M}$, that on input circuits $\vec C$, outputs $Z=(z_1,\ldots,z_M)$ 
, it holds that
$$H^{\epsilon_s}_{\operatorname{min}}(Z | \tilde I) \geq Q_{\operatorname{min}}(n-1) + \log \epsilon_s$$
conditioned on non-aborting event $\Omega$ where $Q_{\operatorname{min}} = \min \{Q : \epsilon_{\mathsf{adv}}(Q,\chi)\footnote{$\epsilon_{\mathsf{adv}}(Q,\chi)$ is monotonically non-decreasing.}  \geq 4\epsilon_s\}$,  $\Pr[\Omega]$ is upper bounded by $\epsilon_\mathsf{adv}(Q,\chi)$.

\begin{center}
  \begin{pchstack}
  \procedure[]{$\mathcal{O}_{\vec C}$($i$):}{
    \codecomment{$q_{\vec C}$ are initialized to 0}\\
    \text{if }q_{\vec C} > Q_{\vec C}\  \ret \bot  \\
    q_{\vec C} = q_{\vec C}+1\\
    \ret \measure(\vec C[i]\ket{0^n})
  }
\end{pchstack}
\end{center}
\end{theorem}

Next, from a string with enough smooth min-entropy, we can effectively extract $\epsilon-$uniform randomness out of it using the strong two-source quantum-proof randomness extractor defined in the following:
\begin{definition}[Two-Source Quantum-Proof Strong Extractor (Lemma 74, \cite{jpmc_cr_2})] 
Let a function $\extract : \{0,1\}^{n_1} \times \{0,1\}^{n_2} \rightarrow\{0,1\}^\ell$ be a quantum-proof strong ($n_1,\kappa_1,n_2,\kappa_2,\ell,\epsilon_{ext}$) two-source extractor. Then for any independent source $\rho_{X_1X_2E}$ classical on $X_i$ with $H_{\operatorname{min}}(X_1|E)_\rho \geq \kappa_1$ and $H^{\epsilon_{s}}_{\operatorname{min}}(X_2|E)_{\rho} \geq \kappa_2 + \log_2(1/\epsilon_2)$ where $\epsilon_2 \in (0,1)$ and $\epsilon_s \in (0,1]$,  we have
$$\norm{\rho_{\extract(X_1,X_2)X_iE } - \tau_{\ell} \otimes \rho_{X_iE}}_{1} \leq 6\epsilon_{s} + 2\epsilon_{ext} + 2\epsilon_2$$ where ${\tau_\ell}$ is a maximally mixed state on $\ell$ bits.
\label{def:randomness_extraction}
\end{definition}

We need strong quantum-proof randomness extractor because the weak source (in our case the provider's output string) is supplied by the (potentially malicious) provider and the adversary may have quantum side-information on it. The second input that is held by the verifier is assumed to be secret and independent from the provider's source. In this paper, we sometimes refer to the extractor's second (private) input as seed\footnote{Note that the strong two-source extractor here allows either one of the sources to be public while the strong seeded extractor only allows the seed to be public.}.

\subsection{Constructive Cryptography (CC)}
In the framework of constructive cryptography (CC) \cite{Maurer2012ConstructC, Maurer2016} (also known as abstract cryptography (AC) \cite{Maurer2011AbstractC}), security is defined as the indistinguishability between two worlds, rather than a game between a challenger and an adversary. The model defines worlds using concepts of resources and converters, analogous to the functionalities and protocols defined in the Universal Composability (UC) model \cite{canetti2001universally}. The CC framework has been used to analyze various cryptographic systems such as DIDComm \cite{Badertscher2024}, quantum key distribution (QKD) \cite{Portmann2022}, ratcheting \cite{Jost2019}, functional encryption \cite{Matt2015} and others \cite{Kerber2021, Portmann2017}. Resources expose interface(s) to which converter(s) can attach to. In the CC framework, the real world is composed of resources composed in parallel. Then, the protocol converter(s) are attached to the resources to convert them into other resources. The ideal world resource is made up of an ideal resource, to which a simulator is attached to. The privacy notion and other properties that should be exhibited by the protocol should be captured in the ideal resource. In the CC framework, we prove the following property
\begin{linenomath}
$$
\pi R \approx_{\epsilon} S\sigma,
$$
\end{linenomath}
where $\pi$ represents the protocol executed on top of a real resource $R$ and $\sigma$ represents a simulator acting on an ideal resource $S$. It translates what the adversary can do in the real world (left side) into what the adversary needs to do in the ideal world (right side) to achieve the same result. Assuming that $\mathsf{REAL} := \pi R $ and $\mathsf{IDEAL} := S\sigma $ expose the same interfaces, we now consider a distinguisher that attaches to the resources and tries to tell them apart. If any distinguisher can only succeed with negligible probability, then we achieve the indistinguishable property as required. Then, we say that $\pi$ constructs $S$ from $R$ denoted as $R \xrightarrow[]{\pi} S$.

\textbf{Composability.} The construction notion briefly discussed above is proven to be composable in the CC framework \cite{Maurer2011AbstractC, Maurer2012ConstructC}. That is,
\begin{linenomath}
$$
R \xrightarrow[]{\pi} S \wedge S \xrightarrow[]{\pi'} T \implies R \xrightarrow{\pi' \circ \pi} T.
$$
\end{linenomath}

\textbf{Syntax.} We follow the resource syntax from \cite{Badertscher2024} when defining the interfaces that a resource exposes. The interfaces are defined by first specifying the party that the interface is exposed to, followed by the command as the first parameter while the rest of the parameters represent the inputs to the interface. For example, the command $\cmdposttx$ for the interface accessible to party $i$ taking input $m$ is defined as $i(\cmdposttx,m)$. We further use the syntax $this$ to address the interface described in the resource itself. 
 
\section{Verifiable Random Sampling (VRS)}
\label{sec:vrs}

In this section, we describe the VRS syntax and security properties at a high level. The high-level overview for the use cases of $\vrs$ is given in \cref{sec:usecases}. In \cref{sec:vrs_with_cc_combined}, we will reconcile the analysis and protocol given in the CC model in \cref{sec:vrs_cc} with the high-level overview given below in \cref{sec:vrs_with_cc_combined}.
First, we define Verifiable Random Sampling (VRS) using similar syntax to VRF as follows:

\begin{definition}
Verifiable random sampling (VRS) is a tuple of algorithms $\vrs_{\mathcal{D}}:= (\gen, \cmdchallenge, \cmdresponse, \eval)$ defined as follows:

The following algorithms are executed by verifier $V$:
\begin{itemize}
 \item $(\idealseed_{chl}, \idealseed_{ext}) \gets \vrs.\gen(\secparam)$: On input a security parameter $\secparam$, outputs a challenge seed $\idealseed_{chl}$ and an extractor seed $\idealseed_{ext}$.
  \item $((u, \sigma)/\bot) \gets \vrs.\eval(\idealseed_{chl},\idealseed_{ext},x,\vec C,\vec z)$: On input a challenge seed $\idealseed_{chl}$, an extractor seed $\idealseed_{ext}$, an input string $x\in \{0,1\}^\ell$, a quantum circuit vector $\vec C$, and an output string $\vec z \in \{0,1\}^n$, outputs a sample $u$ with a proof $\sigma$ or aborts with $\bot$. 
\end{itemize}

The following algorithm is executed by client $T$:
\begin{itemize}
   \item $\vec C \gets \vrs.\cmdchallenge(\idealseed_{chl},x)$: On input a challenge seed $\idealseed_{chl}$ and an input string $x \in \{0, 1\}^\ell$, outputs a circuit vector $\vec C$.
\end{itemize}

The following algorithm is executed by provider $P$:
\begin{itemize}
  \item $\vec z \gets \vrs.\cmdresponse(\vec C)$: On input a quantum circuit vector $\vec C$, outputs a classical output response string $y \in \{0, 1\}^n$.
\end{itemize}
\label{def:vrs}
\end{definition}

For the purpose of showing the security property captured by the ideal resource, we state informally the security properties that VRS should satisfies as follows:
\begin{enumerate}
    \item Correctness: $\eval$ outputs a valid sample $u$ from the distribution $\mathcal{\tilde D}$ when the protocol is executed honestly.
    \item Non-deterministic / Freshness: Given two protocol executions of $\cmdchallenge, \cmdresponse,\cmdeval$ with the same inputs, the resulting outputs $u,u'$ are not the same, $u\neq u'$ except with collision probability caused by the independent sampling from the output distribution.
    \item Information-theoretic Randomness: The output distribution $\mathcal{\tilde D}$ for samples $u$ is statistically close to the ideal distribution $\mathcal{D}$.
\end{enumerate}

We specifically model and proved the case where party $T$ and party $P$ are corrupted. We further discuss the trust model with different corruption patterns in \cref{sec:trust_model_instant}.

\section{Verifiable Random Sampling from Quantum Random Circuit Sampling, $\idealrandom_{\mathcal{D}}^{\idealblockchain}$}
\label{sec:vrs_cc}
\subsection{Overview}
\label{sec:main_overview}

We first list the parties involved and their roles. We refer to \cref{sec:tech_overview} for a high-level overview of our scheme. The symbols used can be found in \cref{tab:symbol_table}.

\textbf{Party.} We model and analyze the case where the verifier $V$ is an honest party while the client $T$ and the provider $P$ can be corrupted by the adversary. We further provide the analysis of different corruption patterns in \cref{sec:trust_model_instant}. The parties are summarized below:
\begin{enumerate}
    \item $P$: Provider with access to the quantum computing resource,
    \item $E/\mathcal{A}$: Adversary,
    \item $T$: Client who interacts with the provider to generate a timestamped transcript,
    \item $V$: Verifier who would check the output of provider and finally outputs the final sample.
\end{enumerate}

\begin{figure}[ht!]
\centering
\makebox[\linewidth][c]{%
\begin{adjustbox}{minipage=1.176\linewidth,scale=0.85, center}
\centering
\begin{pchstack}
\begin{tabular}{c|l}
\hline
Symbol & Meaning \\
\hline
\multicolumn{2}{c}{\textbf{Certified Randomness}}\\
\hline
$\vec C$ & Quantum circuit vector \\
$\vec z$ & Response vector (Classical string) \\
$n, N = 2^n$ & Number of qubits \\
$\xebtest$ & XEB test score \\
\hline
\multicolumn{2}{c}{\textbf{Blockchain}}\\
\hline
$\aux$ & Block auxiliary information \\
$\mathbf{B}$ & Blockchain block \\
$\blockheight$ & Block ID \\
$\transactionhash$ & Transaction ID \\
$\timebound$ & Block time bound \\
$\bcentropybound$ & Block entropy bound \\
$\tau$ & Timestamp \\
\hline
\multicolumn{2}{c}{\textbf{Others}}\\
\hline
$\context$ & Use case specific context\\
$\mathcal{U,D}$ & Uniform and arbitrary distributions \\
\hline
\end{tabular}

\pchspace

\begin{tabular}{c|l}
\hline
Symbol & Meaning \\
\hline
\multicolumn{2}{c}{\textbf{Model}}\\
\hline
$\idealblockchain$ & Blockchain\\
$\idealrandom$ & Verifiable random sampling\\
$\idealcommchannel$ & Communication channel\\
$\idealseed$ & Seed\\
$\pi$ & Protocol (Converter)\\
\hline
\multicolumn{2}{c}{\textbf{Function}}\\
\hline
$\hash()$ & Hash function\\
$H()$ & von Neumann entropy\\
$H_{\operatorname{min}}()$ & Min-entropy\\
$H^{\epsilon}_{\operatorname{min}}()$ & Smooth min-entropy\\
\hline
\multicolumn{2}{c}{\textbf{Parties}}\\
\hline
$T,P,V$ & Client, Provider, Verifier\\
$E, \mathcal{A}$ & Adversary\\
$\mathcal{H}$ & Honest party set\\
$\mathcal{C}$ & Corrupted party set\\
$pub$ & All party set\\
$\mathcal{I}$ & Generic party set\\
\hline
\end{tabular}
\end{pchstack}
\end{adjustbox}
}
\caption{Symbol Table.} 
\label{tab:symbol_table}
\end{figure}

In the rest of this section, we will first introduce various preliminary resources used for modeling the protocol in the context of CC such as blockchain $\idealblockchain$, a read-only storage resource with basic access control called the seed resource $\idealseed$, 
communication channel resource $\idealcommchannel$. Then, we introduce verifiable random sampling as an ideal resource $\vrs_{\mathcal{D}}^{\idealblockchain}$, followed by a protocol ($\pi_{T}^{\vrs}$, $\pi_{P}^{\vrs}$,$\pi_{V}^{\vrs}$) that will construct the ideal resource, and finally we prove the security of our construction.

\subsection{Messaging Channel, $\idealcommchannel$}
\label{sec:ideal_comm_channel_ideal_ro}

Authenticated channels are sufficient for our use cases, and our construction would work over insecure communication channels as well. We use the authenticated messaging channel from \cite{Matt2015} denoted as resource $\idealcommchannel^{I\rightarrow R}$. The authentication property of the messaging channel means that the adversary can neither modify any message nor impersonate either party. Slightly modifying the definition from \cite{Matt2015}, the authenticated channel can be defined as follows:

\begin{definition}[Adapted from \cite{Matt2015}]
An authenticated channel from initiator (I) to receiver (R), denoted by $\idealcommchannel^{I\rightarrow R}$, with an eavesdropper (E), is a resource with three interfaces $I$, $R$, and $E$. On input a message $m$ at interface $I$ using the command $\cmdsend$, the same message can be read at interfaces $R$ and $E$ using the receive command $\cmdreceive$.  
\label{def:communication_channel}
\end{definition}

\subsection{Bulletin Board in Timed Batched Mode with Auxiliary Information, $\idealblockchain$}
\label{sec:idealblockchain_formal}

\begin{figure}
    \begin{adjustbox}{minipage=1.176\linewidth,scale=0.85, center}
    \begin{titledbox}{Resource $\idealblockchain$}
    \begin{pchstack}
        \begin{pcvstack}
            \procedure[]{Init():}{
                \mathbf{\blockheight} := 1 \\
                \mathbf{B^*} := \emptyset \\ 
                \mathbf{B} := [B_0]
            }
            \pcvspace
             \procedure[]{$i(\cmdlatest)$:}{ 
                \codecomment{Latest BlockID}\\
                \ret \blockheight-1
            }
        \end{pcvstack}
      \pchspace
      \procedure[]{$i(\cmdread, id, txid := \emptyset)$:}{
      \codecomment{$i \in \{ P,E,T,V\}$} \\
      \codecomment{Read message}\\
       \text{if } id = -1: \\
       \quad id := \blockheight-1\\
       \text{if } txid = \emptyset: \\
       \quad \ret \mathbf{B}[id] \\
       \text{else:}\\
       \quad \ret \mathbf{B}[id][txid]
      }
    \pchspace
      \procedure[]{$i(\cmdposttx, m)$:}{
      \codecomment{$i \in \{ P,E,T,V\}$} \\
        \transactionhash \gets \hash(\blockheight || m)\\
        \codecomment{Record message}\\
        \mathbf{B^*}[\transactionhash] := m \\
        \ret (\blockheight, \transactionhash)
      }
    \pchspace
        \procedure[]{$M(\cmdmine, \tau, \aux)$:}{
        \codecomment{Finalize Block}\\
        \text{if } \tau - \mathbf{B}[\blockheight-1].\tau \geq \timebound: \\
        \quad \ret \bot \\
        \mathbf{B}[\blockheight] := (\mathbf{B^*}, \tau, \aux, \blockheight) \\
        \mathbf{B^*} := \emptyset\\
        \blockheight := \blockheight + 1\\
        \ret \ok
      }
    \end{pchstack}
    \end{titledbox}
    \end{adjustbox}
    \caption{$\idealblockchain$ - (Ideal) Batched bulletin board resource where $\blockheight$ is the unique sequential block identifier, $\mathbf{B^*}$ is the latest pending block content, and $\mathbf{B}$ contains all historic block information. $\mathbf{B^*}$ is appended to $\mathbf{B}$ when a new block is mined. 
    }
    \label{fig:ideal_blockchain}
\end{figure}

The proposed scheme requires the existence of a public ledger that can record messages and later provide the recorded messages when queried. Our starting point is the bulletin board abstraction provided in \cite{bulletinboard_Choudhuri2017, bulletinboard_Gaddam2023} for the UC model. The bulletin board allows the user to post an arbitrary message to the board and later retrieve it using the associated counter (denoted as $\blockheight$). However, it is insufficient for our scheme, since we additionally require the pinned message to be batched into a time-sensitive sequential ordering where the time interval of each batch of posted messages is upper bounded by some bound $\timebound$. While the notion of time is investigated in \cite{Kiayias2016, Baum2021, LiuZhang2020}, these works are concerned with synchronized clocks for the purpose of synchronous computation in multi-party computation. In particular, the participants in our scheme do not need to have access to a synchronized clock, and it is sufficient for a trustworthy service to timestamp the message in a coarse manner. Combining the aforementioned properties, we give an ideal bulletin board resource with timed batched mode in the context of the CC model.

\textbf{Batched Bulletin Board resource, $\idealblockchain$}. 
We assume that $\blockheight$ is a unique increasing block identifier, and we associate each batch with ``block''. In practice, the blocks are published periodically within some time bound $ \timebound$ with increasing $\blockheight$. $\idealblockchain$ provides timing guarantee such that for all valid $id$, 
\begin{linenomath}
$$
\mathbf{B}[id+1].\tau - \mathbf{B}[id].\tau < \timebound.
$$ 
\end{linenomath}
The timestamp $\tau$ is determined by the block miner. It is generally checked by consensus that the timestamp is reasonable and chronologically consistent. The block miner is considered to be honest, and its honest behavior is encouraged through an incentive mechanism (reward fees) as analyzed in \cite{Buterin2019} and \cite{Kiayias2017}, where the honest behavior is a Nash equilibrium \cite{Nash1950}. Thus, we assume that the miner in this case will always behave honestly.
We further denote the block auxiliary information (including the block header) as $\aux$. In addition, we assume that the block hash contains some (computational) min-entropy lower-bounded by $\bcentropybound$. That is,
\begin{linenomath}
$$
H_{\operatorname{min}}^{\operatorname{comp}}(\hash(\mathbf{B}, \aux)) \geq \bcentropybound,
$$
\end{linenomath}
where $H_{\operatorname{min}}^{\operatorname{comp}}$ is defined to be the HILL computational min-entropy \cite{Hstad1999, Barak2003}.
This assumption is shown to hold for Bitcoin in \cite{bitcoin_entropy}, where it is analyzed that the computational min-entropy of the Bitcoin block header that is mined using a proof-of-work mechanism can be lower-bounded. We further discuss biasable block hashes and personalized randomness expansion in \cref{sec:blockchain_fail}. 
The blockchain resource is given formally in \cref{fig:ideal_blockchain}, where $i$ denotes the interface for blockchain users, and $M$ denotes the interface for a miner. 

\textbf{On the instantiation of $\idealblockchain$ using existing Blockchains.} While it is possible to instantiate batched bulletin board resource $\idealblockchain$ using just a single central trusted party to perform all the computation and bookkeeping, it is often the case that the same functionality instantiated by distributed systems is more desirable due to the decentralization of trust. Blockchain, an append-only ledger that is maintained by ad-hoc communities, is a natural choice for the use case. Proof-of-work blockchains like Bitcoin \cite{nakamoto2008bitcoin} and proof-of-stake blockchains like Ethereum \cite{wood2014ethereum} are two of the most prominent blockchains. The block miners (proposers or validators), who are maintaining the state of the blockchain, are incentivized to produce a consistent state through an inherent reward structure and penalty system for malicious misbehavior. Similar reasoning is used in \cite{bulletinboard_Choudhuri2017} to instantiate Bitcoin/Ethereum blockchain as the bulletin board.

\subsection{Ideal Verifiable Random Sampling Resource $\idealrandom^{\idealblockchain}_{\mathcal{D}}$}

\begin{figure}[h!]
    \begin{adjustbox}{minipage=1.176\linewidth,scale=0.85,center}
    \begin{titledbox}{Resource $\idealrandom^{\idealblockchain}_{\mathcal{D}}$}
    \begin{pchstack}
    \begin{pcvstack}
       \procedure[]{Init():}{
       \codecomment{Variables representing communication}\\
        \mathbf{C, R} := \emptyset
      }
    \pcvspace
    \procedure[]{$E(\cmdreceive)$:}{
         \ret (\mathbf{C,R})
    }
    \pcvspace
    \begin{pchstack}
        \begin{pcvstack}
            \procedure[]{$T(\cmdsend,\mathbf{C'} )$:}{
                \mathbf{C} := \mathbf{C'}\\
                \ret \ok
            }
            \pcvspace
            \procedure[]{$P(\cmdread)$:}{
                \ret \mathbf{C}
            }
        \end{pcvstack}
        \pchspace
        \begin{pcvstack}
            \procedure[]{$P(\cmdsend, \mathbf{R'})$:}{
                \mathbf{R} := \mathbf{R'}\\
                \ret \ok
            }
            \pcvspace
            \procedure[]{$V(\cmdread)$:}{
                \ret \mathbf{R}
            }
        \end{pcvstack}
    \end{pchstack}
      \pcvspace
       \procedure[]{$P(\cmdresponse)$: \codecomment{$P \in \mathcal{H}$}}{
        (\vec C, \blockheight_{\mathbf{C}}, \cdot) := \mathbf{C}\\
        \vec{z} \sim \mathcal{D}_{\vec C} \text{ where } \mathcal{D}_{\vec C}(x) = |\langle x | \vec C \ket{0^{n \cdot M}} |^2\\
        (\blockheight_{\mathbf{R}}, \transactionhash) \gets \idealblockchain.P(\cmdposttx, (\mathbf{C}, \vec z))\\
        \mathbf{R} := (\vec z, \blockheight_{\mathbf{R}}, \transactionhash) \\
        \ret \ok
      }
      \pcvspace
       \procedure[]{$P(\cmdresponse, \mathbf{R'})$: \codecomment{$P \in \mathcal{C}$}}{
        \mathbf{R} := \mathbf{R'} \\
        \ret \ok
      }
    \end{pcvstack}
      \pchspace
      \begin{pcvstack}
      \procedure[]{$T(\cmdchallenge, \context, x)$: \codecomment{$T \in \mathcal{H}$} }{
        (\widehat B, \cdot, \aux,\blockheight_{\mathbf{C}}) \gets \idealblockchain.T(\cmdread, \idealblockchain.T(\cmdlatest)) \\
        \vec C \gets \derive_{Haar(N)^M}(\hash(\hash(\widehat B, \aux)||\context||x)) \\
        \mathbf{C} := (\vec C,\blockheight_{\mathbf{C}}, \context, x) \\
        \ret \ok
      }
      \pcvspace
       \procedure[]{$V(\cmdeval, \context, x)$: \codecomment{ $V \in \mathcal{H}$}}{
        (\vec z, \blockheight_{\mathbf{R}}, \transactionhash) := \mathbf{R} \\
        (\vec C^*, \blockheight_{\mathbf{C}}^*, \context^*, x^*, \vec z^*) \gets \idealblockchain.V(\cmdread, \blockheight_{\mathbf{R}}, \transactionhash)\\
        (\widehat B^*, \cdot, \aux^*, \cdot) \gets \idealblockchain.V(\cmdread, \blockheight_{\mathbf{C}}^*) \\
        \vec C \gets \derive_{Haar(N)^M}(\hash(\hash(\widehat B^*,\aux^*)||\context^*||x^*)) \\
        \text{if } \exists (\blockheight', \transactionhash') \text{\ s.t.\ } (\context^*,x^*) \in \idealblockchain.V(\cmdread, \blockheight', \transactionhash') \\
        \quad \wedge \blockheight' < \blockheight_{\mathbf{R}} \wedge \transactionhash' < \transactionhash: \\
        \quad \ret \bot \codecomment{Abort on earlier reply, detect multi-attempt}\\
        u \sim \mathcal{D}\\
        (\cdot, \blockheight_{min} , \cdot):= \context \\
        \text{if }\ \vec C^* = \vec C \wedge \xebtestbool(\vec C, \vec z^*, \chi)
        \wedge \blockheight_{\mathbf{R}} - \blockheight_{\mathbf{C}}^* \leq \timeboundsec \\
        \quad \wedge \context^* = \context \wedge \blockheight_{\mathbf{C}}^* \geq \blockheight_{min} \wedge x^* = x:\\
        \quad \ret u \\
        \text{else : } \ret \bot
      }
     \pcvspace
        \procedure[]{$T(\cmdchallenge, \mathbf{C}')$: \codecomment{$T \in \mathcal{C}$}}{
            \mathbf{C} := \mathbf{C'} \\
            \ret \ok
        }
       \pcvspace
        \procedure[]{$V(\cmdeval, u)$: \codecomment{$V \in \mathcal{C}$}}{
            \ret u
        }
      \end{pcvstack}

    \end{pchstack}
    \end{titledbox}
    \end{adjustbox}
    \caption{$\idealrandom^{\idealblockchain}_{\mathcal{D}}$ - (Ideal) Verifiable Random Sampling resource where $\vec C =(C_1 \otimes \cdots \otimes C_{M})$. $\mathcal{H}$ is the honest parties, and $\mathcal{C}$ are the corrupted parties. $\blockheight_{\tau}$ is a timing parameter. 
    }
    \label{fig:ideal_random}
\end{figure}

In this section, we describe the ideal VRS resource $\idealrandom^{\idealblockchain}_{\mathcal{D}}$ and explain why it closely models the real-world setting. The ideal resource is parameterized by an internal blockchain resource $\idealblockchain$ and the resulting target sample distribution $\mathcal{D}$. 
The blockchain resource is used to record the response with a timestamp and to provide a string with sufficient min-entropy for the challenge circuit generation as described in \cref{sec:idealblockchain_formal}.
The resource captures the fact that the challenge circuits are derived by combining an input $x$, a context $\context$, and blockchain data rather than by sampling directly from the challenge circuit distribution. This construction guarantees that the challenge circuits can only be known after the hash block header is mined, assuming all other inputs are predetermined, so that provider cannot pre-sample ahead of time. 
In response, the ideal resource samples from the intended response output distribution\footnote{While it might seem that sampling from any source with sufficient entropy that still passes $\xebtestbool$ is sufficient for the ideal resource specification for the honest interface $P(\cmdresponse)$, there might be a distribution where the former holds but it is distinguishable from the intended response output distribution in the real protocol instantiation.}.
Finally, the verifier $V$ interface will only output a sample from $u \sim \mathcal{D}$ if and only if the challenge circuit is correctly derived, the response is given within a reasonable time, and the given provider output string $\vec z$ passes the $\xebtestbool$ test. This ensures a valid sample $u$ is returned if and only if we are sure that the provider output string $\vec z$ contains enough entropy and is derived independently of the sampling extractor's second input. We also make sure to take the earliest response string to avoid any protocol-level oversampling attempt.
Let $\timeboundsec$ denote the maximum tolerable block interval as specified in \cref{fig:ideal_random}.
To apply \cref{theo:smooth_min_entropy_bound_guarantee}, we set $\timeboundsec \cdot \timebound \leq \timeboundtheo$, where  $\timeboundtheo$ is the maximum latency considered by  \cref{theo:smooth_min_entropy_bound_guarantee}.

Let $T$ denote the client, $P$ denote the provider, $V$ denote the verifier, and $E$ denote the adversary.
We denote $pub:=\{T,P,V, E\}$ to be the public set consisting of all parties, $\mathcal{H}$ to be the honest party set $\mathcal{H}:= pub \setminus \mathcal{C}$, and $\mathcal{C} \subseteq pub$ to be the corrupted party set. A complete specification of the ideal resource $\idealrandom$ is given in \cref{fig:ideal_random}. In our use case, the verifier responsible for verifying the provider's output string is assumed to be honest.

\textbf{Communication.} For communication between client $T$, provider $P$, and the verifier $V$, the resource assumes an authenticated communication channel with eavesdropper $E$, as captured by interface commands $\cmdread,\cmdreceive$ for parties $(T,P,V),\allowbreak E$ respectively. The ideal resource is defined in \cref{sec:ideal_comm_channel_ideal_ro} and is used in the protocol construction in the next section.

 \textbf{Seed.} We model the extractor's second input, the seed, as an ideal resource $\idealseed_{\mathcal{I}}$ (see \cref{fig:ideal_seed}) that stores and supplies a seed to any party in the set $\mathcal{I}$. For example, by setting $\mathcal{I}:= \{V\}$, access to the seed is restricted to party $V$, making it a private seed. We defer the discussion of context-dependent seed initialization and instantiation to \cref{sec:trust_model_instant}. 

\textbf{Domain Separator.} 
We model the domain separator as an input $\context$ to the commands $\cmdchallenge$ and $\cmdeval$. This models the requirement that the verifier $V$ runs the verification algorithm and outputs the sample only if both the circuit-generating party $T$ and the verifier $V$ agree on the same context (e.g., session ID, party $T$'s unique identifier, etc.) We further encode the verifier-specified challenge block constraint $\blockheight_{min}$ within the context $\context$, allowing the verifier to specify that the sampling must only occur at or after the block $\blockheight_{min}$.

\begin{figure}
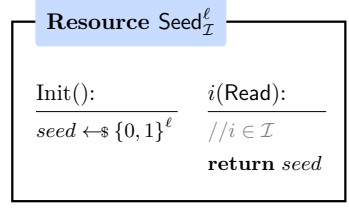

    \centering
    \begin{adjustbox}{minipage=1.176\linewidth,scale=0.85, center}
    \begin{titledbox}{Resource $\idealseed_{\mathcal{I}}^{\ell}$}
    \begin{pchstack}
      \procedure[]{Init():}{
        seed \sample \{0,1\}^{\ell}
      }
      \pchspace
      \procedure[]{$i(\cmdread)$:}{
      \codecomment{$i \in \mathcal{I}$} \\
        \ret seed
      }
    \end{pchstack}
    \end{titledbox}
    \end{adjustbox}
    \caption{$\idealseed_{\mathcal{I}}^{\ell}$ Ideal random seed resource}
    \label{fig:ideal_seed}
\end{figure}

\subsection{Verifiable Random Sampling Protocol Converter $\pi^{\vrs}_{\{T,P,V\}}$}

We now discuss the real resource and the protocol converter. We start with the prominent use case considered during the construction of $\vrs$ in this paper. We discuss and showcase multiple different use cases and potential other instantiations in \cref{sec:usecases}. In addition, we discuss different threat models or corruption patterns for  $\vrs$ in \cref{sec:trust_model_instant}.

\textbf{Use Case.} In practice, a server or an algorithm that needs a verifiable random value from a client party $T$ will play the role of party $V$, for example to determine the lottery winner or to perform leader election. However, client $T$ generally cannot prove that it performs random sampling correctly, but the client $T$ still wants to supply its own sample rather than letting the verifier $V$ pick a sample on its behalf. Client $T$ together with provider $P$ will jointly produce a verifiable sample from a specified distribution with some guaranteed min-entropy. Finally, the verifier $V$ can perform the randomness extraction using the extractor from \cref{def:randomness_extraction}. At this point, the sample can be proven to be a legitimate sample from the uniform distribution. Client $T$ can ensure the last step is performed correctly using techniques from \cref{sec:trust_model_instant}. We give the construction for $\vrs_{\mathcal{U}}^{\idealblockchain}$ as follows:

\textbf{Client $T$.} From the perspective of party $T$, it hashes the latest public blockchain data together with an input $x$ and a context $\context$ to pseudorandomly derive context-specific quantum circuits $\vec C$. $T$ then submits the challenge circuits to $P$. The protocol for party $T$ is given in \cref{fig:protocol_converter_all}.

\textbf{Prover $P$.} The protocol for party $P$ is given in \cref{fig:protocol_converter_all}. The party $P$ protocol is simple. It just needs to evaluate the given quantum circuit vector $\vec C\ket{0^{n \cdot M}}$ and then measure the resulting state. Then, it returns the measured bit-string $\vec z $ back to the requester. The provider $P$ also needs to record its response on-chain using $\idealblockchain.P(\cmdposttx, m)$ so that the response is timestamped for some response message $m$. 

\textbf{Verifier $V$.}  The verifier $V$ tests that the on-chain response is consistent with its view and runs the test $\xebtestbool(\vec C, \vec z)$ to ensure the received sample is a valid sample. $V$ further checks that there does not exist any other earlier response on-chain to prevent post-selection bias. Finally, it runs the randomness extractor on the received sample using a secret seed and outputs the distribution sample given by the randomness extractor. By the randomness extractor theorem, the sample is a sample from an $\epsilon-$close uniform distribution. The concrete protocol is given in \cref{fig:protocol_converter_all}.

\begin{figure}[ht!]
    \begin{adjustbox}{minipage=0.48\linewidth,scale=0.85}
    \hspace*{-\fboxsep}
        \centering
 \begin{titledbox}{Converter $\pi_{T}^{\idealrandom}$}
    \begin{pchstack}
       \procedure[]{$T(\cmdchallenge,\context, x)$:}{
        (\widehat B, \cdot, \aux,\blockheight_{\mathbf{C}}) \gets  \idealblockchain.T(\cmdread, 
        \\ \quad\quad\quad\quad\quad \quad\quad \quad \quad \idealblockchain.T(\cmdlatest)) \\
        \vec C \gets \derive_{Haar(N)^M}(\hash(\hash(\widehat B \\ 
        \quad \quad \quad \quad \quad\quad \quad\quad\quad ,\aux)||\context||x)) \\
        \mathbf{C} := (\vec C,\blockheight_{\mathbf{C}}) \\
        \idealcommchannel^{T \rightarrow P}.T(\cmdsend, \mathbf{C})  \\
        \ret \ok
      }
    \end{pchstack}
    \end{titledbox}
      \begin{titledbox}{Converter $\pi_{P}^{\idealrandom}$}
    \begin{pcvstack}
    \pcvspace
      \procedure[]{$P(\cmdread)$:}{
        \mathbf{C} \gets \idealcommchannel^{T \rightarrow P}.P(\cmdreceive) \\
        \ret \mathbf{C}
      }
    \pcvspace
      \procedure[]{$P(\cmdresponse, \vec z)$:}{
        (\vec C, \blockheight_{\mathbf{C}}) \gets this.P(\cmdread)\\
        \vec z \gets \measure(\vec C \ket{0^{n\cdot M}} \} \\
        (\blockheight_{\mathbf{R}}, \transactionhash) \gets \idealblockchain.P(\cmdposttx, (\mathbf{C}, \vec z)) \\
        \mathbf{R} := (\vec z, \blockheight_{\mathbf{R}}, \transactionhash) \\
        \idealcommchannel^{P \rightarrow V}.P(\cmdsend, \mathbf{R})  \\
        \ret \ok
      }
    \end{pcvstack}
    \end{titledbox}
    \end{adjustbox}
      \begin{adjustbox}{minipage=0.3\linewidth,scale=0.85}
    \hspace*{-\fboxsep}
        \centering
        \begin{titledbox}{Converter $\pi_{V}^{\idealrandom}$}
        \begin{pcvstack}
      \procedure[]{$V(\cmdread)$:}{
      \mathbf{R} \gets \idealcommchannel^{P \rightarrow V}.V(\cmdreceive)  \\
        \ret \mathbf{R}
      }
      \pcvspace
       \procedure[]{$V(\cmdeval, \context, x)$:}{
       (\vec z, \blockheight_{\mathbf{R}}, \transactionhash) \gets this.V(\cmdread) \\
        (\vec C^*, \blockheight_{\mathbf{C}}^*, \context^*, x^*, \vec z^*) \gets \idealblockchain.V(\cmdread, \blockheight_{\mathbf{R}}, \transactionhash)\\
        (\widehat B^*, \cdot, \aux^*, \cdot) \gets \idealblockchain.V(\cmdread, \blockheight_{\mathbf{C}}^*) \\
        \vec C \gets \derive_{Haar(N)^M}(\hash(\hash(\widehat B^*,\aux^*)||\context^*||x^*)) \\
        \text{if } \exists (\blockheight', \transactionhash') \text{\ s.t.\ } (\context^*,x^*) \in \idealblockchain.V(\cmdread, \blockheight', \transactionhash') \\
        \quad \wedge \blockheight' < \blockheight_{\mathbf{R}} \wedge \transactionhash' < \transactionhash: \\
        \quad \ret \bot \codecomment{Abort on earlier reply, detect multi-attempt}\\
        seed_q \gets \idealseed_{V}.V(\cmdread) \\
        u \gets \extract(\vec z, seed_q)\\
       (\cdot, \blockheight_{min} , \cdot):= \context \\
        \text{if }\ \vec C^* = \vec C \wedge \xebtestbool(\vec C, \vec z^*) 
        \wedge \blockheight_{\mathbf{R}} - \blockheight_{\mathbf{C}}^* \leq \timeboundsec \\
        \quad \wedge \context^* = \context \wedge \blockheight_{\mathbf{C}}^* \geq \blockheight_{min} \wedge x^* = x:\\
        \quad \ret u \\
        \text{else : } \ret \bot
      }
    \end{pcvstack}
    \end{titledbox}
    \end{adjustbox}
    \caption{Converter $\pi_{P}^{\idealrandom}$ for emulating interface of $P$, converter $\pi_{T}^{\idealrandom}$ for emulating interface of $T$, and converter $\pi_{V}^{\idealrandom}$ for emulating interface of $V$ of $\idealrandom$}
    \label{fig:protocol_converter_all}
\end{figure}

\subsection{Protocol Summary}
\label{sec:protocol_summary}
The protocol can be summarized as follows: 
\begin{itemize}
    \item Step 0: Verifier $V$ and client $T$ optionally agree on the input $x$ and context $\context$ used in the session. Otherwise, client $T$ picks an input $x$ and a context $\context$ itself. 
    
    \item Step 1: Client $T$ generates circuit $\vec C$ by deriving it from a combination of the latest block hash, input and the context, represented by $T(\cmdchallenge, \context, x)$.

    \item Step 2: Provider $P$ gets the derived circuit $C$, samples from it by evaluating $\vec C\ket{0^{n\cdot M}}$ and then measuring in the computational basis to obtain $\vec z$. $P$ then posts this output. The steps are represented by $P(\cmdresponse)$.

    \item Step 3: Verifier $V$ verifies four things -- no earlier transcript, consistency of circuit derived from block hash, $\vec z$ passes the $\xebtestbool$, and timing constraint. The steps are represented by $V(\cmdeval, \context, x)$.
\end{itemize}

\subsection{Security Analysis for $\idealrandom_{\mathcal{U}}^{\idealblockchain}$}
The following correctness theorem states that if all parties are honest, then the ideal resource $\idealrandom^{\idealblockchain}_{\mathcal{U}}$ where $\mathcal{D}:=\mathcal{U}$ is constructed.

\begin{theorem}
\label{theo:correctness}
    For the protocols $(\pi_{T},\pi_{P}),\pi_{V}$  defined in \cref{fig:protocol_converter_all}, we have that
    $$\pi_{T}\pi_{P}\pi_{V}(\idealblockchain||\idealcommchannel^{T\rightarrow P}||\idealcommchannel^{P\rightarrow V}||\idealseed_{V}) \approx \idealrandom_{\mathcal{U}}^{\idealblockchain}.$$
\end{theorem}

\begin{proof}
The correctness of the proposed construction follows directly from the entropy guarantee from \cref{theo:smooth_min_entropy_bound_guarantee} and randomness extractor \cref{def:randomness_extraction}, and the rest by relabeling and inspection.
\end{proof}

We set $\mathcal{H}:= \{V\}$ and $\mathcal{C}:= \{T,P\}$ which is considered the main use case of our scheme.
We prove the main security theorem of our scheme stated as the theorem below where the verifier $V$ is honest and the provider $P$ together with the client $T$ are corrupted. We rely on composability theorem of CC framework to replace ideal resource with real resource.

\begin{theorem}
\label{theo:security}
    Consider an ideal resource $\idealblockchain$ from \cref{fig:ideal_blockchain}, an ideal resource $\idealcommchannel$ from \cref{def:communication_channel}, an ideal resource $\idealseed$ from \cref{fig:ideal_seed}, and an ideal resource $\vrs_{\mathcal{U}}^{\idealblockchain}$ from \cref{fig:ideal_random}, then there exists a negligible function $\epsilon$ such that

    $$\pi_V(\idealblockchain||\idealcommchannel^{T\rightarrow P}||\idealcommchannel^{P\rightarrow V}||\idealseed_{V})\approx \idealrandom^{\idealblockchain}_{\mathcal{U}}\sigma_{\{P,E,T\}}.$$
\end{theorem}

\begin{proof}
 \textit{Sketch.} The full proof is provided in the \cref{sec:full_sec_proof}. We give the intuition of the proof. From the perspective of the distinguisher, the challenge circuit is generated using the randomness from the latest block hash and there is not enough time to spoof the output to reduce the min-entropy guarantee of the output string while still passing $\xebtestbool$. Similarly, the distinguisher also does not know at the time of response submission to the bulletin board, what seed will ultimately be used to derive the sample. Therefore, if the distinguisher does not have access to this ``seed'' randomness, it cannot distinguish the final random sample from an appropriately defined uniform distribution due to the statistical guarantee of the randomness extractor.   
\end{proof}

\section{Extending $\idealrandom_{\mathcal{U}}^{\idealblockchain}$ and Discussion}
\subsection{$\vrs_{\mathcal{D}}$: VRS for Arbitrary Distribution}
\label{sec:rejection_sampling_error}

Rejection sampling  \cite{von195113} can be used to sample a target distribution $f(x)$ given a source distribution $g(x)$ that is easier to sample from. For example, it is used in lattice-based signatures \cite{lyu12} to sample from a distribution that is independent of the secret. However, existing works require a known source distribution $g(x)$, whereas in our $\vrs$ instantiation, we may only be sampling from a distribution that is only $\epsilon-$close to the intended source distribution because of the property of the randomness extractor (\cref{def:randomness_extraction}). Therefore, we introduce and prove an error-version of the rejection sampling algorithm below.

\noindent\textbf{Rejection Sampling with Error. } Rejection sampling \cite{von195113} performs sampling from an arbitrary target distribution with known probability mass function $f(x)$ given a source distribution $g(x)$ by accepting samples drawn from $g(x)$ with probability $f(x)/(M\cdot g(x))$. Then, the final output distribution is indeed $f(x)$. 
However, in practice, the available source distribution may differ from the ideal: we may only have access to $\tilde g(x) := g(x)+\eta(x)$, where $\eta(x)$ represents the error. Let the support be defined as $\mathsf{supp}(g) := \{x \in \mathcal X : g(x) \neq 0\}.$
We give a rejection sampling using the approximate source distribution $\tilde g(x)$ and prove its properties with respect to a target distribution $\tilde f(x)$  (with error) as follows:

\begin{lemma}[Rejection Sampling with Error]
\label{lemma:rejection_sampling_error}
Let the ideal source PMF be $g$ and ideal target PMF be $f$. Assume that $f(x) \leq M\cdot g(x)$ for all $x$ and $\mathsf{supp}(\tilde g)\subseteq \mathsf{supp}(g)$.
Given an actual source PMF $\tilde g$ where $SD(g,\tilde g)= \epsilon$, the following algorithm, in expected steps at most $(\frac{M}{1-2M\epsilon})$, outputs a sample from a distribution whose PMF is $\tilde f(x)$ where $SD (f,\tilde f) \leq  \frac{2M\epsilon}{1-2M\epsilon}$ if $2M\epsilon < 1$:
\begin{enumerate}
    \item $X \sim \tilde g$.
    \item Output $X$ with probability $\frac{f(X)}{M\cdot g(X)}$. Otherwise, go to step (1).
\end{enumerate}
\label{def:rejection_sampling}
\end{lemma}

\noindent The proof is given in \cref{sec:proof_rej}. We are now ready to lift $\vrs_{\mathcal{U}, \epsilon_u}$ to $\vrs_{\mathcal{D}, \epsilon_d}$ for any $\mathcal{D}$ where $\epsilon_u$ and $\epsilon_d$ are the distribution errors. We assume that the uniform samples suffice, which can be achieved by obtaining a longer uniform sample or repeating the protocol. Then, we apply \cref{lemma:rejection_sampling_error} to obtain a sample from the distribution $\mathcal{\tilde D}$ such that $SD(\mathcal{\tilde D,D})=\epsilon_d \leq \frac{2M\epsilon_{u}}{1-2M\epsilon_u}$. Let acceptance probability be $p := \Pr[accept]$. To ensure a negligible failure probability $(1-p)^\alpha \leq 2^{-\lambda}$, it is sufficient to set the repetition factor $\alpha > (1/p) \lambda \ln 2$.

\subsection{Different Corruption Patterns.}
Our $\vrs_{\mathcal{U}}^{\idealblockchain}$ is analyzed for when $T$ and $P$ are malicious and $V$ is honest. Assuming the provider is always malicious, we provide analysis for different corruption patterns below:
\begin{enumerate}
    \item (Honest $T,V$; Malicious $P$): $\vrs_{\mathcal{U}}^{\idealblockchain}$ also captures this setting since the honest strategy is contained within arbitrary malicious strategies. 
    \item (Honest $T$; Malicious $P,V$): Since the verifier is dishonest, we consider the case that there exists a third party who is interested in the existence of a sample $u$ given by the verifier $V$. In this case, $T$ is that third party. First, the verifier uses the future block hash as the extractor's second input to extract the final sample together with provider's response string $\vec z$. The third party (in this case $T$) simply re-runs the verification algorithm but with the second input replaced as mentioned above.
    \item (; Malicious $T,P,V$): Similar to the above, except the third party is neither $T,P,$ nor $V$.
\end{enumerate}

{
\setlength{\textfloatsep}{0pt}   
\setlength{\floatsep}{0pt}       
\setlength{\intextsep}{0pt}      
\begin{table}[ht]
\rowcolors{1}{}{lightgray}
\begin{adjustbox}{minipage=\linewidth,scale=1}
\centering
\begin{tabular}{r|r|r|r}
  \hline
  & Client $T$ & Provider $P$ & Verifier $V$\\
  \hline
  Captured by $\vrs_{\mathcal{U}}^{\idealblockchain}$ & \Warning &  \Warning & \cmark \\
  Captured by $\vrs_{\mathcal{U}}^{\idealblockchain}$ & \cmark & \Warning & \cmark \\
  Item 2 Above & \cmark & \Warning & \Warning \\
  Item 3 Above & \Warning & \Warning & \Warning \\
   \hline
\end{tabular}
\end{adjustbox}
\caption{Corruption Patterns. Provider is always malicious. \cmark$\ $ is an honest party while \Warning$\ $ is a corrupted party.} 
\label{tab:corrupt}
\end{table}
}

\noindent\textbf{Further Discussion.} We defer discussion on the transition to fault-tolerant QC, performance, domain separation, and other topics to the Appendix.

\section{Conclusion}
In this work, we introduced verifiable random sampling (VRS), a new protocol which allows for publicly verifiable sampling from any target distribution in a composable manner. The protocol outputs information-theoretic randomness and is resistant to collusion. In addition, we prove security in the constructive cryptography (CC) model, ensuring composability. Building on experimentally demonstrated quantum certified randomness protocols, VRS is practical for deployment in today's blockchain environments. The versatility of VRS is demonstrated through several use cases, including smart contract randomness and e-lotteries, applicable to blockchain. We believe that our protocol addresses a growing need for high quality randomness that is both unbiasable and verifiable. 

While our model is built on top of RCS, we leave it to future work to adapt our construction by basing the entropy source on preimage sampling in an NP-search problem. This regime is more relevant for fault-tolerant quantum computers of the future, and using them can make our construction more efficient. Additionally, in this work we analyze the security of the model against a restricted class of adversaries. We believe this restriction can be lifted, and we leave this as an important direction for future work.

\section*{Acknowledgments}
We thank Pradeep Niroula for helpful discussions of random quantum circuit sampling and its applications. We thank Rob Otter for the executive support of the work and invaluable feedback on this project. 
The authors thank their colleagues at the Global Technology Applied Research center of JPMorganChase for their support.

\section*{Disclaimer}
This paper was prepared for informational purposes with contributions from the Global Technology Applied Research center of JPMorgan Chase \& Co. This paper is not a product of the Research Department of JPMorgan Chase \& Co. or its affiliates. Neither JPMorgan Chase \& Co. nor any of its affiliates makes any explicit or implied representation or warranty and none of them accept any liability in connection with this paper, including, without limitation, with respect to the completeness, accuracy, or reliability of the information contained herein and the potential legal, compliance, tax, or accounting effects thereof. This document is not intended as investment research or investment advice, or as a recommendation, offer, or solicitation for the purchase or sale of any security, financial instrument, financial product or service, or to be used in any way for evaluating the merits of participating in any transaction.
 
\bibliographystyle{plainurl}
\bibliography{main}

\appendix

\section{Deferred Discussion}
\label{sec:discussion}

\subsection{Reconciliating VRS specified in \cref{sec:vrs_cc} with VRS Overview from \cref{sec:vrs}}
\label{sec:vrs_with_cc_combined}

\textbf{Security Properties.} It can be seen that the correctness property informally stated in \cref{sec:vrs} is captured by \cref{theo:correctness}. The non-deterministic randomness property is captured by the fact that the $V$-interface from the ideal resource given in \cref{fig:ideal_random} always outputs a ``fresh'' sample from the intended distribution as output when given $x$ as input. For non-deterministic sample generation within the same block epoch and the same input $x$, we rely on the domain separation $\context$ where $\context$ is the domain separating string as discussed in \cref{sec:domain_seperation}.
The information theoretic randomness property is captured by the fact that the $V$-interface outputs random sample from the intended output distribution. 
\\
\noindent\textbf{Algorithm Mapping.} $\vrs.\gen$ is the setup algorithm that generates the necessary challenge seed and extractor seed. We assume the setup is done correctly and model the extractor seed as an ideal resource $\idealseed$ which specifies the exact parties that have access to the $\idealseed$ resource. We further discuss different setup instantiations for seed in \cref{sec:trust_model_instant}. The challenge seed $\idealseed_{chl}$ is derived using hash function on a blockchain hash, a context $\context$, and an input $x$. 
$\vrs.\cmdchallenge$, $\vrs.\cmdresponse$, and $\vrs.\cmdeval$ are modeled in the ideal resource $\idealrandom^{\idealblockchain}_{\mathcal{D}}$ given in \cref{fig:ideal_random} with some minor modifications. 

\subsection{Different Instantiation and Setup}
\label{sec:trust_model_instant}

\textbf{`Private' Seed Setup. } The seed (extractor's second input) is required to be unknown to the provider at the time of the response submission. We specifically model the seed as a verifier's private seed to achieve this. For use cases where there does not exist private storage for the verifier $V$, we can make use of the future unpredictable block hash as a second entropy source for the randomness extractor. We use this approach for smart contract use cases as described in \cref{sec:usecases}. If the verifier is composed of distributed entities, then we can make use of common MPC techniques such as distributed key generation (DKG) \cite{Gennaro2006} to derive a reconstructible shared common seed as described in \cref{sec:usecases}.

\subsection{Transitioning from Today's Quantum Computing Devices to Future Fault-Tolerant Quantum Computing Devices}
\label{sec:nisq_to_perfect}

Our $\VRS_\mathcal{U}^{\idealblockchain}$ is constructed using RCS-based challenge-response as the entropy source. RCS-based challenge-response works over today's noisy QC device as demonstrated in \cite{jpmc_cr,jpmc_cr_2} experimentally with complexity-theoretic foundation from \cite{aaronson2023certified}. However, this approach is time-sensitive whereby we have to bound the provider's response time and we have to set the difficulty of the problem to be not too hard as the verifier needs to compute the $\xebtestbool$ which is considered computationally expensive. 
Fortunately, our $\vrs$ construction can be made more efficient by basing the entropy source on sampling pre-image in a NP-search problem that is solvable by quantum polynomial-time machines but not
by classical probabilistic polynomial-time machines as studied in \cite{Yamakawa2024}. However, it requires operations that are not feasible on today's QC devices. Assuming there will be fault-tolerant quantum computing devices in the future, we can instantiate $\vrs$ based on the aforementioned pre-image sampling problem instead. Notably, the framework developed in this work can be modified from being RCS-based challenge-response to the pre-image search challenge-response.

\subsection{Domain Separation}
\label{sec:domain_seperation}
Domain separation is encouraged in hash functions where instead of computing $\hash(x)$, we compute $\hash(\context||x)$ instead for some context string $\context$ \cite{Bellare2020}. 
Similarly, we have that the $\vrs$ protocol also takes in as input a context $\context$ where $\context$ can represent client-verifier public keys, session ID, and many other use-cases-specific contexts. The $\context$ can be deterministic and has no entropy but it has to be unique for the context it operates in. For example, we can pass in the smart contract address as context to derive challenge circuits for a specific smart contract $V$. Context enables different challenge circuits to be derived deterministically for different client contexts $\context$ in the same block epoch.

\subsection{Instantiation Analysis}
Experiments on the RCS-based certified randomness protocol are reported in \cite{jpmc_cr, jpmc_cr_2}. In particular, \cite{jpmc_cr_2} derives concrete parameter settings and estimates the expected verification cost in GPU hours where the GPU hours are with respect to  Intel Data Center GPU Max with approximately $46$ TeraFLOPS theoretical peak performance. In that work, the effective latency is modeled as $T_M / n_{parallel}$, where $T_M$ denotes the time between the moment the circuit (or the measurement basis) is revealed and the moment the response is received considering the base unit $n_{parallel}=1$, and $n_{parallel}$ is the number of quantum-computing resources available to the prover. In a blockchain setting, $T_M$ is likely to remain relatively fixed, but the prover can increase $n_{parallel}$ to handle more challenge circuits simultaneously, thereby reducing the effective latency per unit circuit. If this effective latency becomes sufficiently small, the validation cost could drop below one GPU hour. Since validation is fully parallelizable, validation time on the order of minutes is achievable through straightforward parallelization.

\subsection{Advantage of VRS in the Malicious Client Setting}
\label{sec:vrs_advantage_over_vrf}
Let us work in a setting in which the client is rationally malicious, attempting to maximize value for themselves, and the provider is honest. Now consider a transformation of our VRS protocol which replaces the circuit generator and quantum server with a classical VRF. Within an epoch of time $T$, the output of this transformed protocol is precisely determined by the inputs submitted by the client. If a client submits the same input twice within this epoch of period $T$, the output is identical. Therefore, a client can perform a grinding attack by submitting a set $S$ of inputs and choosing $x_{\text{max}} \in S$ such that
$$\text{Utility}\left(\text{Output}(\text{x}_{\text{max}})\right) = \max\limits_{y \in S} \text{Utility}\left(\text{Output}(y)\right),$$
where $\text{Output}$ is a map defined as the VRF output for a given input, and $\text{Utility}$ is a well-defined map on possible client inputs during the epoch quantifying the utility of the output for the client. Note that $\text{Output}$ is well-defined for our protocol under replacement of the circuit generator and server with a classical VRF, by the determinism of a classical VRF. In our proposed VRS (with the circuit generator and quantum server as originally proposed), there is no well-defined analogue to $\text{Output}$ due to freshness. Freshness ensures a resubmission of the same input yields different outputs with high probability. Therefore, while the client can grind and obtain circuits a priori within an epoch of time $T$, they cannot effectively grind for optimal output as they can under replacement with a classical VRF.

\subsection{Others}
\label{sec:blockchain_fail}
\textbf{When Blockchain Fails.} We further provide analysis for when the assumption that blockchain provides entropy fails. We assume the blockchain hash contains enough computational entropy \cite{bitcoin_entropy} to argue about the pseudorandomness of the circuits generated. For the case that the blockchain entropy is slightly biasable, the block hash entropy is reduced by a few bits. Our scheme will remain secure under this few-bit biasability scenario. 

\textbf{Why not just use existing blockchain entropy? } 
Blockchain entropy can be interpreted as a source of randomness from a randomness beacon. The randomness is broadcast at a relatively fixed interval. $\vrs$ provides a way to verifiably ``expand'' the number of samples extractable from a single-source of randomness. Importantly, the samples extracted will be different for different parties $\context$ even though they rely on the same source of randomness without any extra entropy source from their end. Importantly, the $\vrs$ protocol converts any source of \textit{pseudo}-randomness into information-theoretic randomness.

\textbf{Provider's Authentication.} The provider might be rewarded based on the quality of the response it returns. Note that our scheme implicitly captures the authentication aspect for the response it submits by requiring the provider to submit to the blockchain through signing the response as a transaction. The verification algorithm can simply, in addition, check that the response is submitted by the provider it requested using the provider's public key.

\section{Deferred Preliminaries}

\begin{definition}[VRF]
\label{def:vrf}

Verifiable random function (VRF) is a tuple of algorithms $\VRF := (\VRFgen{\cdot}, \VRFeval{\cdot}, \VRFvfy{\cdot})$ defined as follows:
\begin{itemize}
 \item $(\SK,\PK) \gets \VRFgen{\secparam}$: On input a security parameter $\secparam$, outputs a secret key $\SK$ and public key $\PK$.
    \item $(y,\pi) \gets \VRFeval{\SK,x}$: On input a secret key $\SK$ and an input string $x \in \{0, 1\}^\ell$, outputs $y \in \{0, 1\}^n$ and a proof $\pi$.
    \item $b \gets \VRFvfy{\PK,x,y,\pi}$: On input a public key $\PK$, an input string $x$, an output string $y$, and a proof $\pi$, outputs a bit $b \in \{0, 1\}$,
\end{itemize}

VRF satisfies the following properties:

\begin{enumerate}
    \item \textbf{Correctness} - $\forall \lambda \in \mathbb{N}$, $\forall (\SK, \PK)$ in the image of $\VRFgen{1^\lambda}$, $\forall x$, and  $\forall (y, \pi)$ in the image of $\VRFeval{\SK, x}$, 
    \begin{linenomath}
    $$
        \VRFvfy{\PK, x, y, \pi} = 1.
    $$
    \end{linenomath}
    \item \textbf{Unique Provability} - $\forall \PK$ (not necessarily generated by $\VRFgen{}$, $\forall x \in \{0, 1\}^L$, $\forall y_0, y_1 \in S$, and all possible proofs $\pi_0$, $\pi_1$, we have 
    \begin{linenomath}
    $$
        (\VRFvfy{\PK, x, y_0, \pi_0} = \VRFvfy{\PK, x, y_1, \pi_1} = 1)  \implies y_0 = y_1.
    $$
    \end{linenomath}
    In other words, $\forall x \in \{0, 1\}^L$, $\exists$ a valid proof for at most one $y$ value.
    \item \textbf{Pseudorandomness} -
    $\forall \operatorname{PPT}$ adversaries $\mathcal{A} = (\mathcal{A}_1, \mathcal{A}_2)$, we have 
    \begin{linenomath}
    $$
        \operatorname{Adv}^{\operatorname{vrf}}_{\mathcal{A}}(\lambda) \coloneqq \left\vert \operatorname{Pr}\left[\operatorname{Exp}^{\operatorname{vrf}}_{\mathcal{A}}(\lambda)=1\right] - \frac{1}{2} \right\vert \leq \operatorname{negl}(\lambda).
    $$
    \end{linenomath}
\end{enumerate}

\begin{algorithm}
    \caption{The VRF Security Experiment : $\operatorname{Exp}^{\operatorname{vrf}}_{\mathcal{A}}(\lambda)$}
    \begin{algorithmic}[1]
        \State $b \gets^{\$} \{0,1\}$
        \State $(pk, sk) \gets^{\$} \mathrm{VRF.Gen}(1^\lambda)$
        \State $(x^\star, \text{state}) \gets \mathcal{A}_1^{\mathcal{O}_{\mathrm{eval}}(\cdot)}(pk)$
        \State $(y_0, \pi) \gets \mathrm{VRF.Eval}(sk, x^\star)$
        \State $y_1 \gets_R S$
        \State $b' \gets \mathcal{A}_2^{\mathcal{O}_{\mathrm{eval}}(\cdot)}(\text{state}, y_b)$
        \State \Return $b == b'$
    \end{algorithmic}
    \label{alg:vrf_sec_exp}
\end{algorithm}

\begin{algorithm}
    \caption{The VRF Evaluation Oracle : $\mathcal{O}_{\mathrm{eval}}(x)$}
    \begin{algorithmic}[1]
    \If{$x = x^\star$}
        \State \Return $\bot$
    \EndIf
    \State $(y, \pi) \gets \mathrm{VRF.Eval}(sk, x)$
    \Return $(y, \pi)$
    \end{algorithmic}
    \label{alg:vrf_eval_oracle}
\end{algorithm}
\end{definition}

\section{Complexity of RCS}
\label{sec:hardness}

We first recall the definition for (Mixed) Linear Cross-Entropy Benchmark $\mlxeb_{b,k}(\mathcal{D})$ metric from \cite{jpmc_cr} for some distribution $\mathcal{D}$.
\begin{definition}[$\mlxeb_{b,k}(\mathcal{D})$ \cite{jpmc_cr}]
\label{def:mlxeb}
Let $\mathcal{D}$ be a probability distribution over quantum circuits on n qubits. Given $\vec C := (C_1, \dots, C_k)$ drawn from $\mathcal{D}^k$, output sample $\vec z =(z_1,\dots,z_k) \in (\{0,1\}^n)^k$ such that 
$$\frac{1}{k} \sum^{k}_{i=1} p_{C_i}(z_i) \geq \frac{b}{N}.$$
\end{definition}
We usually consider the case that $\mathcal{D} \sim Haar(N)$.

We now recall the complexity class needed for the paper.
\begin{definition}[Quantum classical Arthur-Merlin ($\qcam$)]
Let $x$ be a statement and $L$ be a language.
The complexity class $\qcam$ consist of languages $L$ for which there exists a  quantum polynomial-time algorithm $V$ (Arthur) and a polynomial $p$ such that the following hold:
\begin{itemize}
    \item If $x\in L$, then there exists a polynomial $q$ and a polynomial-length classical string $w \in \{0,1\}^{q|x|}$ such that $\Pr_{r\in \{0,1\}^{p|x|}}[V(x,r,w) \ accepts] \geq 2/3$.
    \item If $x\notin L$, then for every a polynomial $q$ and every polynomial-length classical string $w \in \{0,1\}^{q|x|}$ such that $\Pr_{r\in \{0,1\}^{p|x|}}[V(x,r,w) \ accepts] \leq 1/3$.
\end{itemize}
\end{definition}
$\qcamtime(T)/q(A)$ is the generalization of $\qcam$ whereby the verifier can use running time of $T$ and receives $A$ bits of quantum advice that only depends on $n$.

The following problem that is conjectured to be hard in \cite{aaronson2023certified}. The same assumption is used to demonstrate quantum supremacy \cite{Preskill2012} through random circuit sampling (RCS) in the original Google quantum supremacy experiment \cite{Arute2019}.
\begin{definition}[Long List Quantum Supremacy Verification ($\llqsv$)]
Given a list of $M=O(N^{3})$ circuit-string tuples $\{(C_i, s_i) : i \in [M]\}$ such that $C_i \sim \mathcal{D}$, distinguish the cases:
\begin{itemize}
    \item Yes-case: for each $i \in [M]$, $s_i$ is sampled from $C_i$.
    \item No-case: for each $i \in [M]$, $s_i$ is sampled uniformly at random and thus is independent of $C_i$.
\end{itemize}
\end{definition}

In addition, it is proved that the following assumption hold in the random oracle model, i.e. $\llqsv(\mathcal{D}) \notin \qcamtime(2^Bn^{O(1)})/q(2^{B}n^{O(1)})$.
\begin{assumption}[Long List Hardness Assumption $(\llha_{B}(\mathcal{D})$ \cite{aaronson2023certified} ] There exists no $\qcam$ protocol for which quantum Arthur solves $\llqsv$ in time $2^Bn^{O(1)}$ given access to a quantum advice of length $2^Bn^{O(1)}$.
\end{assumption}

\begin{theorem}[\cite{aaronson2023certified} Theorem 5.21 ]
\label{theo:llqsv_hardness}
Relative to a random oracle, $\llqsv \notin \qcamtime(c^n)/q(c^n)$ for some constant $1<c<2$.
\end{theorem}

We now state the result from \cite{jpmc_cr} that says passing $\mlxeb$ with low von Neumann entropy means solving $\llqsv$, in other words, $\llqsv_B(\mathcal{D}) \in \qcamtime(2^Bn^{O(1)})/O(n)$. 

\begin{theorem}[Passing MLXEB test with low entropy solves $\llqsv$(Theorem 8, \cite{jpmc_cr})]
There exists a quantum-classical Arthur-Merlin protocol which on input an $O(n)$-bit advice string solves $\llqsv_B(\mathcal{D})$ which means $\llqsv_B{(\mathcal{D})}\in \qcamtime(2^Bn^{O(1)})/O(n)$, if there exists a device $\mathcal{A}$ which runs in polynomial time and satisfies the following:
\begin{itemize}
    \item  $\mathcal{A}$ solves $\mlxeb$ with probability $q=\Pr_{\vec C \sim \mathcal{D}^k, \vec z \sim \mathcal{A}(\vec C)}\left[\sum_{i=1}^k p_{C_i}(z_i) \geq \frac{bk}{N}\right]$,
    \item $H(Z | \vec C )_{\mathcal{A}} < \frac{B}{2} \left(\frac{bq-1-\epsilon}{b-1}\right)$ where $\epsilon = n^{-O(1)}$.
\end{itemize}
\end{theorem}

As a direct corollary, we have the following theorem that bound the entropy obtained assuming $\llha$ hardness assumption hold.
\begin{theorem}[Entropy Guarantee, Theorem 6 \cite{jpmc_cr}]
\label{theo:entropy_guarantee_unused}
Assuming that $\llha_{B}(\mathcal{D})$ for distribution $\mathcal{D}$ over quantum circuits acting on $n$ qubits, given any device that on input takes $k$ independently sampled circuits $\vec C \sim \mathcal{D}^k$ outputs a classical string $Z:=(z_1||\dots||z_k)\in {\{0,1\}}^{nk}$ solving $\mlxeb_{b,k}$ with probability q, it holds that $$H(Z | \vec C) \geq \frac{B}{2}\left(\frac{bq-1}{b-1}-n^{-\omega(1)}\right).$$
\end{theorem}

Therefore, assuming that $\llha_{B}(\mathcal{D})$ for distribution $\mathcal{D}$ over quantum circuits acting on $n$ qubits, we can derive entropy from the output string from the algorithm. We can further derive smooth min-entropy lower bounded by the von Neumann entropy by repetition using Entropy Accumulation Theorem (EAT) \cite{Dupuis2020}.

\section{Deferred Proofs}
\label{sec:proof_rej}
\begin{proof}
We will now prove \cref{lemma:rejection_sampling_error}.
Let $\tilde g(x) = g(x)+\eta(x)$ where $\frac{1}{2}\sum_x|\eta(x)|=\epsilon$ such that $SD(g,g+\eta)=\epsilon$.
We compute probability of acceptance first.
\begin{align*}
    \Pr[accept]&= \Pr[x=0]  \cdot \Pr[accept |x=0] + \cdots \tag*{(LoTP)}\\
    &=\sum_x ( g(x)+\eta(x)) \cdot \frac{f(x)}{M \cdot g(x)}\\
    &=\frac{1}{M}\left(1+\sum_x \left(\eta(x)\cdot \frac{f(x)}{g(x)}\right)\right)\\
    &=\frac{1}{M}(1+ \delta) \tag*{(\ensuremath{\delta = \sum_x (\eta(x)\cdot \frac{f(x)}{g(x)})})}\\
    &\leq \frac{1}{M}(1+ |\delta|)\\
    &\leq \frac{1}{M}\left(1+{\sum_x \abs{\eta(x)}\cdot \frac{f(x)}{g(x)}} \right) \tag*{(tri.ineq)}  \\
    &\leq \frac{1}{M}(1+\sum_x \abs{\eta(x)}M)  \tag*{\ensuremath{(\frac{f(x)}{g(x)} \leq M)}} \\
    &\leq \frac{1}{M}+2\epsilon,
\end{align*}
and $\Pr[accept]=\frac{1}{M}(1+ \delta) \geq \frac{1}{M}(1 - |\delta|) \geq \frac{1}{M}(1 - 2M\epsilon)$. Note that $\abs{\delta} \leq 2M\epsilon<1$. Thus,
\begin{align*}
\tilde f (x) &= \Pr[X=x | accept]\\
&= \Pr[X=x]\cdot \Pr[accept|X=x] \cdot 1/\Pr[accept]\\
&= ( g(x)+\eta(x)) \cdot \frac{f(x)}{M \cdot g(x)}\cdot 1/\left(\frac{1}{M}(1+\delta)\right)\\
&= \frac{f(x)}{1+\delta} \cdot \left(1+\frac{\eta(x)}{g(x)}\right).
\end{align*}
Finally, the error is given by
\begin{align*}
    \tilde f(x)-f(x) &= \frac{f(x)}{1+\delta}\left(1+\frac{\eta(x)}{g(x)}\right)-f(x)\\
    &= \frac{f(x)}{1+\delta} \left(\frac{\eta(x)}{g(x)}-\delta\right).
\end{align*}
The statistical difference is then
\begin{align*}
    SD(\tilde f, f) &= \frac{1}{2}\sum_x \abs{ \tilde f(x)- f(x)}\\
    &= \frac{1}{2}\sum_x \abs{\frac{f(x)}{1+\delta} \left(\frac{\eta(x)}{g(x)}-\delta\right)}\\
    &= \frac{1}{2}\sum_x \frac{\abs{f(x)\left(\frac{\eta(x)}{g(x)}-\delta\right)}}{\abs{{1+\delta}}}\\
    &= \frac{1}{2{\abs{{1+\delta}}}} \sum_x \abs{f(x)\left(\frac{\eta(x)}{g(x)}-\delta\right)}\\
    &\leq  \frac{1}{2{\abs{{1+\delta}}}} \left(\sum_x \abs{f(x)\left(\frac{\eta(x)}{g(x)}\right)} + \sum_x \abs{f(x) \delta}\right) \tag*{(tri.ineq)}\\
    &\leq  \frac{1}{2{\abs{{1+\delta}}}} \left(\sum_x{ \left(\abs{\eta(x)}\cdot \frac{f(x)}{g(x)}\right)} + |\delta|\right)\\
    &\leq  \frac{1}{2{\abs{{1+\delta}}}}(2M\epsilon + 2M\epsilon)\\
    &\leq \frac{1}{{\abs{{1+\delta}}}}(2M\epsilon)\\
    &\leq \frac{1}{{\abs{{1-|\delta|}}}}(2M\epsilon) \tag*{(reverse tri.ineq, \ensuremath{\abs{\delta} < 1})}\\
    &\leq \frac{1}{{{1-2M\epsilon}}}(2M\epsilon)
\end{align*}
\end{proof}

\section{Deferred Security Proof}
\label{sec:full_sec_proof}

\begin{proof}
We will prove \cref{theo:security} by providing a simulator $\simulator$ and a sequence of hybrids $\hybrid_0,\cdots, \hybrid_4$ such that for any (classical\footnote{The result should naturally extends to the quantum distinguisher if we use the min-entropy from complexity guarantee stated at \cref{theo:entropy_guarantee_unused}.}) probabilistic polynomial-time distinguisher $\mathbf{D}$:
    $$\pi_V(\idealblockchain||\idealcommchannel^{T\rightarrow P}||\idealcommchannel^{P\rightarrow V}||\idealseed_{V})\approx\idealrandom^{\idealblockchain}_{\mathcal{U}}\sigma_{\{P,E,T\}}.$$

We start by specifying the simulator $\simulator_{\{P,E,T\}}$ given in \cref{fig:simulator_main}. The simulator is a straightforward proxy to the internal interface since we do not make use of any secret in the ideal resource. We give the distinguisher $\mathbf{D}$ access to $\mathcal{O_{\vec C}}$ that operates exactly like the oracle defined in \cref{theo:smooth_min_entropy_bound_guarantee}. Note that the query count $Q_{\vec C}$ is tracked separately for each different $\vec C$ corresponding to a new session. Then, we proceed to show a sequence of indistinguishable hybrids. 
\\

\noindent\textbf{Hybrid $\hybrid_0$}. This hybrid is the same as $\pi_V(\idealblockchain||\idealcommchannel^{T\rightarrow P}||\idealcommchannel^{P\rightarrow V}||\idealseed_{V})$.
\\

\noindent\textbf{Hybrid $\hybrid_1$}. This hybrid is the same as previous one except we abort when there is a collision of circuits provided by in 
$\derive_{Haar(N)^M}$ 
in any of the previous session. The collision probability is negligible. 
\\

\noindent\textbf{Hybrid $\hybrid_2$}. This hybrid is the same as previous one except we abort when the verification $\xebtestbool(\vec C, \vec z)$ passes but the $Q_{\vec C} < Q_{\operatorname{min}}$ where $Q_{\operatorname{min}}$ is defined as in \cref{theo:smooth_min_entropy_bound_guarantee}. This hybrid transition make use of the fact that adversary cannot spoof more than specified or else it fails with overwhelming probability. Since the protocol matches that of \cref{theo:smooth_min_entropy_bound_guarantee}, we have that the abort condition is negligible. That is the verification passes with at most negligible probability $\Pr[\Omega] = 4\epsilon_s$ when $Q_{\vec C} < Q_{\operatorname{min}}$ and $\epsilon_s$ is negligible. 
\\

\noindent\textbf{Hybrid $\hybrid_3$}. This hybrid is the same as $\hybrid_2$ except that instead of sampling circuit $u$ using extractor on the response recorded on-chain with the second extractor input provided by $\idealseed_V$, we sample $\widehat u$ uniformly from $\mathcal{U}$. Since $u$ is the output from randomness extractor defined in \cref{def:randomness_extraction} and assume that $\vec z$ fulfills $H_{\operatorname{min}}^{\epsilon}$ and $seed_q$ fulfills $H_{\operatorname{min}}$ requirement, we have that the distinguishing advantage between the state where $u$ is returned and the state where $\widehat u$ is returned is upper bounded by the trace distance \footnote{Trace distance $T(\rho,\sigma):=\frac{1}{2}\norm{\rho-\sigma}_1$ quantifies the distinguishing advantage between two quantum states $\rho$ and $\sigma$.} $\frac{1}{2} (6\epsilon_s+2\epsilon_{ext}+2\epsilon_2)$ which is negligible.
Since extractor second input trivially satisfies the min-entropy requirement $\ell_{seed} > \kappa_1$, we now show that the smooth min-entropy requirement for $\vec z$ is satisfied. We consider the case that the distinguisher $\mathcal{D}$ has the exact computational capacity specified in \cref{sec:crypto_assump} and has access to the quantum circuit query oracle $\mathcal{O}_{\vec C}$.

Assuming $\idealblockchain$ bulletin board entropy $\bcentropybound$ is sufficiently high, the block hash $\hash(\widehat B || \aux)$ contains enough entropy and the distinguisher cannot guess it in advance except querying it at the latest block where it is first published. Since the distinguisher do not have enough time to completely spoof the output sample against pseudorandomly but freshly generated challenge circuit, it is forced to returns valid quantum sample in some of the output string $Z = \vec z$. 

Let $\blockheight_{\vec C}$ be the block where challenge circuit is derived from and $\blockheight_{J}$ be the block where the response can be submitted to just in time for any $Q$-query bounded algorithm that can only spoof at most $M-Q_{\vec C}$ samples that still passes the $\xebtestbool$ when rest of $Q_{\vec C}$ samples come from ideal quantum circuit evaluation (response from $\mathcal{O}_{\vec C}$). From the block timing guarantee defined in \cref{sec:idealblockchain_formal} and the aforementioned conditions, by ensuring $\idealblockchain$ time bound $\timebound \leq \timeboundtheo$, we have that 
$$ \mathbf{B}[\blockheight_{J}].\tau - \mathbf{B}[\blockheight_{\vec C}].\tau \leq \tau_{res}.$$

Since there is at least $Q_{\operatorname{min}}$ quantum samples because of the abort condition in the previous hybrid and considering the scenario described matches exactly that of assumption made in \cref{theo:smooth_min_entropy_bound_guarantee}, then using \cref{assum:rcs_unstructure} and \cref{theo:smooth_min_entropy_bound_guarantee}, it can be concluded that 
$$H^{\epsilon_s}_{\operatorname{min}}(Z | \tilde I) \geq Q_{\operatorname{min}}(n-1) + \log \epsilon_s.$$
Assume that $Q_{\operatorname{min}}(n-1) + \log \epsilon_s \geq \kappa_2 + \log_2(1/\epsilon_2)$
, we conclude that the change is indistinguishable.
\\

\noindent\textbf{Hybrid $\hybrid_4$}. This hybrid is the same as $\hybrid_3$ except we re-label the communication channel and the seed to the symbol used in the ideal resource with simulator $\idealrandom^{\idealblockchain}_{\mathcal{U}}\sigma_{P\cup E}$. 
Therefore, $\hybrid_4\equiv\hybrid_3$.
\\

At this point with $\hybrid_4$, conditioned on non-aborting events, it can be seen that the distribution of transcript when interacting with $\hybrid_4$ is identically distributed as interacting with $\idealrandom^{\idealblockchain}_{\mathcal{U}}\sigma_{P\cup E}$. 

\end{proof}

\begin{figure}[h!]
    \centering
      \begin{titledbox}{Simulator $\sigma_{\{P,E,T\}}^{}$}
    \begin{pchstack}
    \begin{pcvstack}
    \pcvspace
    \procedure[]{Emulating $E(\cmdreceive)$:}{
         \ret \call E(\cmdreceive)
    }
          \pcvspace
    \procedure[]{Emulating $ P(\cmdread)$: \codecomment{$P \in \mathcal{C}$}}{
        \ret \call P(\cmdread)
    }
      \pcvspace
       \procedure[]{Emulating $P(\cmdresponse, \mathbf{R'})$: \codecomment{$P \in \mathcal{C}$}}{
        \ret \call P(\cmdsend, \mathbf{R'})
      }
    \end{pcvstack}
      \pchspace
      \begin{pcvstack}
        \procedure[]{Emulating $T(\cmdchallenge, \mathbf{C}')$: \codecomment{$T \in \mathcal{C}$}}{
            \ret \call T(\cmdsend, \mathbf{C'}) 
        }
       \pcvspace
        \procedure[]{Emulating $T(\cmdsend,\mathbf{C'} )$:  \codecomment{$T \in \mathcal{C}$}}{
            \ret \call T(\cmdsend, \mathbf{C'})
        }
        \pcvspace
        \procedure[]{Emulating $P(\cmdsend, \mathbf{R'})$:  \codecomment{$P \in \mathcal{C}$}}{
            \ret \call P(\cmdsend, \mathbf{R'})
        }
      \end{pcvstack}
    \end{pchstack}
    \end{titledbox}
    \caption{Simulator $\sigma_{\{P,E,T\}}^{}$ for $\vrs$. 
    }
    \label{fig:simulator_main}
\end{figure}

\end{document}